\include{BoxedEPS}

\documentclass[10pt]{amsart}
\usepackage{graphicx, amsmath,amsthm}

\def\QED{\quad\vbox{\hrule \hbox{\vrule\kern3pt
      \vbox{\kern3pt{}\kern3pt}\kern3pt\vrule}\hrule}}

\newcommand{\RR}{{\mathbb{R}}}

\newcommand{\ZZ}{{\mathbb{Z}}}

\newtheorem{Tm}{Theorem}[section]
\newtheorem{Lm}[Tm]{Lemma}
\newtheorem{Cr}[Tm]{Corollary}

\newtheorem{pr}[Tm]{Proposition}

\numberwithin{equation}{section}

\begin{document}

\bibliographystyle{plain}

\title [Fractional Laplacian and Sparse Stochastic Processes]
{Left-Inverses of Fractional Laplacian and Sparse Stochastic Processes}


 \author{Qiyu Sun}
\address [Qiyu Sun]{Department of
 Mathematics,  University of Central Florida,
Orlando, FL 32816, USA}
\email{qsun@mail.ucf.edu}

\author{Michael Unser}
\address [Michael Unser]{Biomedical Imaging Group, \'Ecole
Polytechnique  F\'ed\'erale de Lausanne, Lausanne 1015, Switzerland}
\email{michael.unser@epfl.ch}

\maketitle


\date{\today}




\begin{abstract}
The fractional Laplacian $(-\triangle)^{\gamma/2}$ commutes with
the primary coordination transformations in the Euclidean space
$\RR^d$:
  dilation, translation and rotation,  and has tight link to splines, fractals and stable Levy processes.
   For $0<\gamma<d$, its inverse is the classical  Riesz potential $I_\gamma$
which is
 dilation-invariant and translation-invariant.
 In this work, we investigate the functional properties (continuity, decay and invertibility) of an extended class of differential operators that share those invariance properties.
In particular, we extend the definition of  the classical  Riesz
potential $I_\gamma$ to any non-integer number $\gamma$
 larger than $d$ and show
that it is the unique left-inverse of the fractional Laplacian
$(-\triangle)^{\gamma/2}$ which is
 dilation-invariant and translation-invariant.
  We observe that,  for any $1\le p\le \infty$ and $\gamma\ge d(1-1/p)$,
there exists a Schwartz function $f$ such that $I_\gamma f$ is not
$p$-integrable. We then introduce the  new unique left-inverse
$I_{\gamma, p}$ of  the fractional Laplacian
$(-\triangle)^{\gamma/2}$ with the property that $I_{\gamma, p}$
is dilation-invariant (but not translation-invariant)  and that
$I_{\gamma, p}f$ is $p$-integrable for any Schwartz function $f$.
We  finally apply that  linear operator $I_{\gamma, p}$ with $p=1$
to solve the stochastic  partial differential equation
$(-\triangle)^{\gamma/2} \Phi=w$ with  white Poisson noise as its
driving term $w$.

\end{abstract}
\maketitle


\section{Introduction}
\setcounter{equation}{0}

  Define the Fourier transform   ${\mathcal F}\! f$ (or $\hat f$ for brevity) of an integrable function $f$ on  the $d$-dimensional Euclidean space $\RR^d$
  by
\begin{equation}\label{fouriertransform.def}
{\mathcal F}\! f(\xi):=\int_{\RR^d} e^{-i \langle {\bf x}, \xi\rangle} f({\bf x}) d{\bf x},\end{equation}
and extend the above definition to  all tempered distributions  as usual.
Here we denote by
$\langle \cdot, \cdot\rangle$ and   $\|\cdot\|$ the standard inner product and norm  on  $\RR^d$ respectively.

Let ${\mathcal S}:={\mathcal S}(\RR^d)$ be the space of all Schwartz functions on $\RR^d$ and ${\mathcal S}':={\mathcal S}'(\RR^d)$
the space of all tempered distributions on $\RR^d$.
For $\gamma>0$, define the {\em fractional Laplacian} $(-\triangle)^{\gamma/2}$
by
\begin{equation}\label{fractionallaplacian.def}
{\mathcal F}((-\triangle)^{\gamma/2}f)(\xi):=\|\xi\|^\gamma \ {\mathcal F} f(\xi),\quad  f\in {\mathcal S}.
\end{equation}
The fractional Laplacian has the remarkable property of being dilation-invariant. It plays a crucial role in the definition of thin plate splines \cite{duchon1977},
is intimately tied to fractal stochastic processes (e.g., fractional Brownian fields) \cite{mandelbrot1968,tafti2009} and
stable Levy processes \cite{chen10}, and  has been used in the study of singular obstacle problems
\cite{caffarelli08, silvestre06}.

In this paper, we present a detailed mathematical investigation of the functional properties of dilation-invariant differential operators together with a characterization of their inverses.
Our primary motivation is to provide a rigorous operator framework for solving the
stochastic  partial differential equation
\begin{equation}\label{randompde.def}
(-\triangle)^{\gamma/2} \Phi=w\end{equation}
with white noise $w$ as its driving term. We will show that this is feasible via the specification of a novel family of dilation-invariant left-inverses
of the fractional Laplacian $(-\triangle)^{\gamma/2}$ which have appropriate $L^p$-boundedness properties.

 We say that a continuous linear  operator $I$ from ${\mathcal S}$ to
 ${\mathcal S}'$ is {\em dilation-invariant}
if there exists a  real number  $\gamma$ such that
\begin{equation}
I (\delta_t f)= t^\gamma \delta_t (If)\quad{\rm for\ all} \ f\in {\mathcal S}\ {\rm  and} \  t>0,
\end{equation}
and {\em translation-invariant} if
\begin{equation}I (\tau_{{\bf x}_0} f)= \tau_{{\bf x}_0} (If)\quad {\rm  for\ all} \ f\in {\mathcal S}\ {\rm  and} \ {\bf x}_0\in \RR^d, \end{equation}
where  the {\em dilation operator} $\delta_t, t>0$  and the {\em translation operator} $\tau_{{\bf x}_0}, {\bf x}_0\in {\mathbb R}^d$ are defined by
$(\delta_t f) ({\bf x})=  f(t {\bf x})$  and
$\tau_{{\bf x}_0} f({\bf x})= f({\bf x}-{\bf x}_0), f\in {\mathcal S}$, respectively.
One may verify that the fractional Laplacian $(-\triangle)^{\gamma/2}, \gamma>0$, is dilation-invariant
and translation-invariant, a central property used in the definition of thin plate splines \cite{duchon1977}.

Next, we define the {\em Riesz potential} $I_\gamma$ (\cite{riesz49}) by
\begin{equation} \label{rieszpotential.eq1}
I_\gamma f({\bf x})=\pi^{-d/2} 2^{-\gamma} \frac{\Gamma((d-\gamma)/2)}{\Gamma(\gamma/2)}\int_{{\mathbb R}^d} \|{\bf x}-{\bf y}\|^{\gamma-d} f({\bf y}) d{\bf y}, \quad  f\in {\mathcal S},
\end{equation}
where $0<\gamma<d$.
 Here the Gamma function $\Gamma$ is given by
 $\Gamma(z)=\int_0^\infty t^{z-1} e^{-t} dt$ when the real part  ${\rm Re}\ z$ is
positive, and is extended analytically to a meromorphic function
on the complex plane. For any Schwartz function $f$, $I_\gamma f$ is continuous and   satisfies
 \begin{equation} \label{rieszpotential.eq2}
 |I_\gamma f({\bf x})|\le C_\epsilon \Big(\sup_{{\bf z}\in {\mathbb R}^d} |f({\bf z})| (1+\|{\bf z}\|)^{d+\epsilon}\Big)  (1+\|{\bf x}\|)^{\gamma-d}\quad {\rm for  \ all}\  {\bf x}\in {\mathbb R}^d,
 \end{equation}
 where $\epsilon>0$ and  $C_\epsilon$ is a positive constant, see also Theorem \ref{generalizedrieszomega1.tm}. Then the Riesz potential $I_\gamma$ is a continuous linear operator from ${\mathcal S}$ to ${\mathcal S}'$.
 Moreover one may verify that $I_\gamma$ is dilation-invariant and translation-invariant,  and also
  that $I_\gamma, 0<\gamma<d$, is  the inverse of the fractional Laplacian $(-\triangle)^{\gamma/2}$; i.e.,
 \begin{equation}\label{rieszpotential.eq3}
 I_\gamma (-\triangle)^{\gamma/2} f=(-\triangle)^{\gamma/2} I_\gamma f=f\quad {\rm  for\ all} \ f\in {\mathcal S}
 \end{equation}
  because
 \begin{equation}\label{rieszpotential.eq4}
 {\mathcal F}(I_\gamma f)(\xi)= \|\xi\|^{-\gamma} {\mathcal F}f (\xi), \ f\in {\mathcal S}.
 \end{equation}

A natural question then is as follows:

{\bf Question 1}: {\em For any $\gamma>0$, is there  a continuous linear operator $I$ from ${\mathcal S}$ to ${\mathcal S}'$
that  is translation-invariant and dilation-invariant,  and that  is an inverse of the  fractional Laplacian $(-\triangle)^{\gamma/2}$?}

 In the first  result of this paper (Theorem \ref{generalizedriesz.tm}), we give an affirmative answer to
the above existence question for
all  positive non-integer numbers $\gamma$ with the invertibility replaced by the left-invertibility, and further prove the uniqueness
of such a  continuous linear  operator.

To state that result, we recall some notation and definitions.
Denote the dual pair  between a Schwartz function  and a tempered
distribution using angle bracket $\langle \cdot, \cdot\rangle$, which  is given by
$\langle f, g\rangle=\int_{{\mathbb R}^d} f({\bf x}) {g({\bf x})} d{\bf x}$ when $f, g\in {\mathcal S}$
(we remark that the dual pair between two complex-valued square-integrable functions
is different from their standard inner product).
 A tempered distribution $f$ is said to be {\em homogeneous of degree $\gamma$} if
$\langle f, \delta_t g\rangle= t^{-\gamma-d} \langle f, g\rangle$ for all Schwartz functions  $g$
 and  all positive numbers $t$. We notice that the multiplier $\|\xi\|^{-\gamma}$ in the Riesz potential  $I_\gamma$, see \eqref{rieszpotential.eq4},
   is a homogenous function of degree $-\gamma\in (-d, 0)$.
 This observation inspires us to follow the definition of homogeneous tempered distribution in \cite{hormanderbook} and then to extend the definition of the Riesz potential  $I_\gamma$ to any non-integer number $\gamma>d$ as follows:
\begin{eqnarray}\label{generalizedriesz.def}
I_\gamma f({\bf x})& := & \frac{(2\pi)^{-d}\Gamma(d-\gamma)}{ \Gamma(d+k_0-\gamma)} \int_{S^{d-1}}  \int_0^\infty r^{k_0-\gamma+d-1}\nonumber  \\
& &
\times   \Big(-\frac{d}{dr}\Big)^{k_0} \Big(e^{ir\langle {\bf x}, \xi'\rangle} \hat f(r\xi')\Big) dr d\sigma(\xi'),
\quad f\in {\mathcal S},
\end{eqnarray}
where  $S^{n-1}=\{\xi'\in \RR^d: \ \|\xi'\|=1\}$ is the unit
sphere in ${\RR}^d$, $d\sigma$ is the area element on $S^{n-1}$, and
 $k_0$ is a nonnegative integer larger than $\gamma-d$.
  Integration by parts shows that the above definition \eqref{generalizedriesz.def}
 of  $I_\gamma f$ is independent on the nonnegative integer $k_0$ as long as it is larger than $\gamma-d$, and
also  that it coincides with the classical Riesz potential when $0<\gamma<d$ by letting $k_0=0$
  and  recalling that
  the inverse Fourier transform ${\mathcal F}^{-1} f$ of an integrable function $f$ is given by
\begin{equation}\label{inversefouriertransform.def}
{\mathcal F}^{-1} f({\bf x}):=(2\pi)^{-d} \int_{\RR^d} e^{i \langle {\bf x}, \xi\rangle} f(\xi) d\xi.\end{equation}
  Because of  the above consistency of definition, we call
 the continuous linear operator $I_\gamma, \gamma\in (0, \infty)\backslash (\ZZ_++d)$ in \eqref{generalizedriesz.def}  the {\em generalized Riesz potential}, where $\ZZ_+$ is the set of all nonnegative integers.

\begin{Tm} \label{generalizedriesz.tm}
Let $\gamma$ be a positive number with $\gamma-d\not\in \ZZ_+$, and let $I_\gamma$ be
the linear operator  defined by \eqref{generalizedriesz.def}.
Then $I_\gamma$ is the {\bf unique} continuous linear  operator from ${\mathcal S}$ to ${\mathcal S}'$  that
 is  dilation-invariant and translation-invariant, and that  is a left inverse of the fractional Laplacian $(-\triangle)^{\gamma/2}$.
\end{Tm}

Let $L^p:=L^p({\mathbb R}^d), 1\le p\le \infty$, be the space of all $p$-integrable functions on ${\mathbb R}^d$ with
the standard norm $\|\cdot\|_p$. The Hardy-Littlewood-Sobolev fractional integration theorem (\cite{steinbook})
says that the Riesz potential $I_\gamma$ is a bounded linear operator from $L^q$ to $L^p$ when $1<p\le \infty, 0<\gamma<d(1-1/p)$
and $q=pd/(d+\gamma p)$. Hence
 $I_\gamma f\in L^p$ for any Schwartz function $f$
  when $0<\gamma<d(1-1/p)$.
We observe that  for any non-integer number $\gamma$ larger than or equal to  $d(1-1/p)$,
  there exists a  Schwartz function $f$ such that $I_\gamma f\not\in L^p$, see  Corollary \ref{nonintegrable.cr}.
An implication of this negative result, which will become clearer in the sequel (cf. Section 4), is that we cannot generally use the translation-invariant inverse $I_\gamma$ to solve
 the stochastic partial differential
 equation \eqref{randompde.def}.
What is required instead is a special left-inverse of the fractional Laplacian
that is dilation-invariant and $p$-integrable. Square-integrability  in particular ($p=2$) is a strict requirement when the driving noise is Gaussian and has been considered in prior work \cite{tafti2009}; it leads to a fractional Brownian field solution, which is the multi-dimensional extension of Mandelbrot's celebrated fractional Brownian motion \cite{blu, mandelbrot1968}.
Our desire to extend this method of solution for non-Gaussian brands of noise leads to the second question. 

{\bf Question 2}: {\em  Let $1\le p\le \infty$ and $\gamma>0$.
 Is there a continuous linear  operator  $I$ from ${\mathcal S}$ to ${\mathcal S}'$
 that  is  dilation-invariant and  a left-inverse of the fractional Laplacian $(-\triangle)^{\gamma/2}$ such that $If\in L^p$ for all Schwartz functions $f$?}

 In  the second   result of this paper (Theorem \ref{integrablefractionalderivative.tm}), we give an affirmative answer to the above question when both $\gamma$ and $\gamma-d(1-1/p)$
 are not integers, and show the uniqueness of such a continuous linear operator.

To
state  that  result, we introduce some additional multi-integer notation.
For ${\bf x}=(x_1, \ldots, x_d)\in {\mathbb R}^d$ and
 ${\bf j}=(j_1, \ldots, j_d)\in {\mathbb Z}_+^d$ (the $d$-copies of the set ${\mathbb Z}_+$),
 we set
$|{\bf j}|:=|j_1|+\cdots+|j_d|$, ${\bf j}!:=j_1!\cdots j_d!$ with $0!:=1$,
 ${\bf x}^{\bf j}:=x_1^{j_1}\cdots x_d^{j_d}$ and $\partial^{\bf j} f({\bf x}):=\partial^{j_1}_{x_1}\cdots\partial^{j_d}_{x_d} f({\bf x})$.
For $1\le p\le \infty$ and $\gamma>0$, we define the linear operator $I_{\gamma, p}$
from ${\mathcal S}$ to ${\mathcal S}'$ with the help of the Fourier transform:
\begin{equation}\label{fractionalderivative.veryolddef}
{\mathcal F}(I_{\gamma, p} f)(\xi)=\Big({\mathcal F} f(\xi)-\sum_{|{\bf j}|\le \gamma-d(1-1/p)}
\frac{\partial^{\bf j} ({\mathcal F} f)({\bf 0})}{{\bf j}!} \xi^{\bf j}\Big) \|\xi\|^{-\gamma}, \quad f\in {\mathcal S},
\end{equation}
which is the natural $L^p$ extension of the fractional integral operator that was introduced in \cite{blu, tafti2009,taftu2010} for $p=2$  and $\gamma\not\in \ZZ/2$.

We call  $I_{\gamma, p}$ the {\em $p$-integrable Riesz potential of degree $\gamma$}, or the
{\em integrable Riesz potential} for brevity. Indeed,
 when both $\gamma$ and $\gamma-d(1-1/p)$ are non-integers, the linear operator
$I_{\gamma, p}$ is the unique
 left-inverse of the fractional Laplacian $(-\triangle)^{\gamma/2}$ that enjoys the
 following dilation-invariance and stability properties.

\begin{Tm}\label{integrablefractionalderivative.tm}
Let  $1\le p\le \infty$, and $\gamma$ is a positive  number such that
both $\gamma$ and $\gamma-d+d/p$ are not nonnegative integers.
 Then $I_{\gamma, p}$ in \eqref{fractionalderivative.veryolddef} is the {\bf unique}  dilation-invariant left-inverse of the fractional Laplacian $(-\triangle)^{\gamma/2}$ such that
its image of the Schwartz space ${\mathcal S}$ is contained in $L^p$.
\end{Tm}

One of the primary application of the $p$-integrable Riesz potentials
is the construction of generalized random processes by suitable functional integration of white noise \cite{tafti2009, taftu2010, Unser2009}.
These processes are defined by the
stochastic partial differential equation \eqref{randompde.def}, the motivation being that the solution should essentially display the same invariance properties as the defining operator (fractional Laplacian).
In particular, these processes will exhibit some level of self-similarity (fractality) because $I_{\gamma, p}$ is dilation-invariant. However, they will in general not be stationary because the requirement for a stable inverse excludes translation invariance. It is this last aspect that deviates from the classical theory of stochastic processes and requires the type of mathematical safeguards that are provided in this paper. While the case of a white Gaussian noise excitation is fairly well understood  \cite{tafti2009}, it is not yet so when the driving term is impulse Poisson noise which leads to the specification of sparse stochastic processes with a finite rate of innovation. The current status has been to use the operator $I_{\gamma,2}$ to specify sparse processes with the restriction that the impulse amplitude distribution must be symmetric \cite[Theorem 2]{Unser2009}. Our present contribution is to show that one can lift this restriction by considering the operator $I_{\gamma,1}$, which is  the proper inverse to handle general impulsive Poisson noise.

To state our third result,
 we recall some concepts about generalized random processes and  Poisson  noises.
Let ${\mathcal D}$ be the space of all compactly supported $C^\infty$ functions with standard topology.
A  {\em generalized random process}
 is a random functional $\Phi$ defined on ${\mathcal D}$
(i.e., a random variable $\Phi(f)$ associated with every $f\in {\mathcal D}$)
which is linear, continuous and compatible \cite{gelfandbook}. 

The white Poisson noise
\begin{equation}\label{whitepoisson.def}
 w({\bf x}):=\sum_{k\in {\mathbb Z}} a_k \delta({\bf x}-{\bf x}_k)\end{equation}
is a generalized random process such that the  random variable associated with a function $f\in {\mathcal D}$ is given by
 \begin{equation}w(f):=\sum_{k\in \ZZ} a_k f({\bf x}_{k}),\end{equation}
 where the $a_k$'s are i.i.d. random variables with probability distribution $P(a)$, and where  the ${\bf x}_k$'s
are random point locations in ${\mathbb R}^n$ which are mutually independent and follow a spatial Poisson distribution
with Poisson parameter $\lambda>0$.
The
 random point locations ${\bf x}_k$ in ${\mathbb R}^n$  follow a {\em spatial Poisson distribution}
with Poisson parameter $\lambda>0$ meaning that for any measurable set $E$ with finite Lebesgue measure $|E|$,
the probability of observing $n$ events in $E$
 (i.e., the cardinality of the set $\{k| \ {\bf x}_k\in E\}$ is equal to $n$)
 is $\exp(-\lambda |E|) (\lambda |E|)^n/ n!$.  Thus, the Poisson parameter $\lambda$
  represents the average number of random impulses per unit.

As the white Poisson noise $w$ is a generalized random process, the stochastic partial differential equation \eqref{randompde.def}
can be interpreted as the following:
\begin{equation}\label{randompde.def2}
\langle \Phi, (-\triangle)^{\gamma/2} f\rangle=\langle w, f\rangle\quad {\rm for\ all} \ f\in {\mathcal D}.
\end{equation}
So if $I$ is a left-inverse of the fractional Laplacian operator $(-\triangle)^{\gamma/2}$, then
\begin{equation}\Phi=I^* w\end{equation} is
{\em literally} the solution of the stochastic partial differential equation \eqref{randompde.def} as
\begin{equation}
\langle I^*w, (-\triangle)^{\gamma/2} f\rangle=\langle w, I(-\triangle)^{\gamma/2} f\rangle=\langle w,  f\rangle \quad {\rm for \ all} \ f\in {\mathcal D},
\end{equation}
where $I^*$ is the conjugate operator of the continuous linear operator $I$ from ${\mathcal S}$ to ${\mathcal S}'$
 defined by $$\langle I^*f, g\rangle:=\langle f, Ig\rangle\quad {\rm for\ all} \  f, g\in {\mathcal S}.$$
The above observation is usable only if we can specify a left-inverse (or equivalently we can impose appropriate boundary condition)
so that $I^*w$ defines a bona fide generalized random process in the sense of Gelfand and Vilenkin; mathematically, the latter is equivalent to providing its characteristic functional by  the Minlos-Bochner Theorem (cf. Section 4).
The following result establishes
that $P_\gamma w:=I_{\gamma, 1}^*w$  is a proper
 solution of the stochastic partial differential equation
\eqref{randompde.def}, where $w$ is the 
Poisson noise defined by (\ref{whitepoisson.def}).

\begin{Tm}\label{generalizedpoisson.tm}
Let $\gamma$ be a positive non-integer   number,  $\lambda$ be a positive number,
$P(a)$ be a probability distribution  with $\int_{\RR} |a| dP(a)<\infty$, and $I_{\gamma, 1}$ be defined as in \eqref{fractionalderivative.veryolddef}.
For any $f\in {\mathcal D}$, define the random variable $P_\gamma w$ associated with $f$ by
\begin{equation}\label{generalizedpoisson.tm.eq1}
P_\gamma w (f):=\sum_{k} a_k I_{\gamma, 1}(f)({\bf x}_k) 
\end{equation}
where  
the $a_k$'s are i.i.d. random variables with probability distribution $P(a)$, and  the ${\bf x}_k$'s
are random point locations in ${\mathbb R}^n$ which are mutually independent and follow a
spatial Poisson distribution with Poisson parameter $\lambda$.
  Then $P_{\gamma} w$ is  
  the generalized random process associated with the characteristic
  functional
\begin{equation}\label{generalizedpoisson.tm.eq2}
  {\mathcal Z}_{P_\gamma w}(f)=   \exp\Big(\lambda \int_{\RR^d}\int_{\RR} \big(e^{-ia (I_{\gamma, 1}f)({\bf x})}-1\big)  dP(a)d{\bf x}\Big),
  \quad  f\in {\mathcal D}.
  \end{equation}
 \end{Tm}

\bigskip

The organization of the paper is as follows.  In  Section \ref{grp.section}, we first introduce
a  linear operator $J_\Omega$  for any
homogeneous  function $\Omega\in C^\infty(\RR^d\backslash \{\bf 0\})$ of degree $-\gamma$, where $\gamma-d\not\in \ZZ_+$.
 The linear operator  $J_\Omega$
becomes the generalized Riesz potential $I_\gamma$ in \eqref{generalizedriesz.def} when $\Omega(\xi)=\|\xi\|^{-\gamma}$;
 conversely, any  derivative of the generalized Riesz potential $I_\gamma$
  is a linear operator $J_\Omega$ associated with some homogeneous function $\Omega$:
$$\partial^{\bf j} I_\gamma f= J_{\Omega_{\bf j}} f\quad {\rm  for\ all} \ f\in {\mathcal S}\ {\rm and} \ {\bf j}\in \ZZ_+^d,$$ where $\Omega_{\bf j}(\xi)=(i\xi)^{\bf j} \|\xi\|^{-\gamma}$.
 We then  study various properties of the above linear operator $J_\Omega$, such as polynomial decay property, dilation-invariance,
translation-invariance,  left-invertibility, and non-integrability in the spatial domain and in the Fourier domain.
The proof of   Theorem  \ref{generalizedriesz.tm} is given at the end of Section \ref{grp.section}.

In Section \ref{irp.section}, we  introduce a linear operator $U_{\Omega,p}$ for any homogeneous function
$\Omega\in C^\infty(\RR^d\backslash \{\bf 0\})$ of degree $-\gamma$, where $1\le p\le \infty$. The above linear operator $U_{\Omega, p}$ becomes
the operator $I_{\gamma,p}$ in \eqref{fractionalderivative.veryolddef} when $\Omega(\xi)=\|\xi\|^{-\gamma}$, and
the operator $J_\Omega$ in \eqref{fractionalderivative.def} when   $0<\gamma<d(1-1/p)$.
We show
that the linear operator $U_{\Omega,p}$ is dilation-invariant, translation-variant and $p$-integrable, and
is a left-inverse of the fractional Laplacian $(-\triangle)^{\gamma/2}$ when $\Omega(\xi)=\|\xi\|^{-\gamma}$.
  The proof of Theorem \ref{integrablefractionalderivative.tm} is given at the end of Section \ref{irp.section}.

In Section \ref{poisson.section}, we give the proof of Theorem \ref{generalizedpoisson.tm} and show that the generalized
random process $P_\gamma w$ can be evaluated pointwise in the sense that
 we can replace
the function $f$ in \eqref{generalizedpoisson.tm.eq1} by the delta functional $\delta$.

In this paper, the capital letter $C$ denotes an absolute positive constant which may vary depending on the occurrence.

\section{Generalized Riesz Potentials}\label{grp.section}

Let  $\gamma$ be a real number such that $\gamma-d\not\in \ZZ_+$, and let
 $\Omega\in C^\infty(\RR^d\backslash \{\bf 0\})$ be a homogeneous function of degree $-\gamma$.
 Following the definition of homogenous tempered distributions in \cite{hormanderbook},
  we define the linear operator $J_\Omega$ from ${\mathcal S}$ to ${\mathcal S}'$
   by
\begin{eqnarray}\label{fractionalderivative.def}
J_\Omega f({\bf x})\!\! & := & \!\! \frac{(2\pi)^{-d}\Gamma(d-\gamma)}{ \Gamma(d+k_0-\gamma)} \int_{S^{d-1}}  \int_0^\infty \Omega(\xi')r^{k_0-\gamma+d-1}\nonumber  \\
& &
\times   \Big(-\frac{d}{dr}\Big)^{k_0} \Big(e^{ir\langle {\bf x}, \xi'\rangle} \hat f(r\xi')\Big) dr d\sigma(\xi'), \quad f\in {\mathcal S},
\end{eqnarray}
where $S^{n-1}=\{\xi'\in \RR^d: \ \|\xi'\|=1\}$ is the unit
sphere in ${\RR}^d$, $d\sigma$ is the area element on $S^{n-1}$, and
 $k_0$ is a nonnegative integer larger than $\gamma-d$.

Note that the linear operator $J_\Omega$ in  \eqref{fractionalderivative.def}
becomes  the generalized Riesz potential  $I_\gamma$ in \eqref{generalizedriesz.def}
when $\Omega(\xi)=\|\xi\|^{-\gamma}$ and $\gamma>0$.
Therefore we call the  linear operator $J_\Omega$ in \eqref{fractionalderivative.def} {\em the generalized Riesz potential  associated with
the homogeneous function $\Omega$ of degree $-\gamma$}, or {\em the generalized Riesz potential} for brevity.

The above definition
 of the generalized Riesz potential $J_\Omega$ is independent on the nonnegative integer $k_0$ as long as it satisfies $k_0>\gamma-d$,
 that can be shown by integration by parts.
Then,  for  $\gamma \in (-\infty, d)$, we may take $k_0=0$ and  reformulate  \eqref{fractionalderivative.def} as
 follows:
   \begin{equation}\label{fractionalderivative.neweq2}
  J_\Omega f({\bf x})=(2\pi)^{-d} \int_{\RR^d} e^{i\langle {\bf x}, \xi\rangle} \Omega(\xi) \hat f(\xi) d\xi
  \quad {\rm for \ all} \ f\in {\mathcal S},
  \end{equation}
  or equivalently
\begin{equation}
\widehat{J_\Omega f}(\xi)= \Omega(\xi) \hat f(\xi)\quad {\rm   for  \ all} \ f\in {\mathcal S},\end{equation}
so that the role of the homogeneous function $\Omega(\xi)$ in \eqref{fractionalderivative.def}
is essentially that of the Fourier symbol for a conventional translation-invariant operator.

Let ${\mathcal S}_\infty$ be the space of all Schwartz functions $f$ such that
$ \partial^{\bf i}\hat f({\bf 0})=0$  for all ${\bf i}\in \ZZ_+^d$, or equivalently that
$\int_{{\mathbb R}^d} {\bf x}^{\bf j} f({\bf x}) d{\bf x}=0$ for all ${\bf j}\in \ZZ_+^d$.
Given a  homogenous function $\Omega\in C^\infty({\mathbb R}^d\backslash \{{\bf 0}\})$, define
the linear  operator $i_\Omega$ on ${\mathcal S}_\infty$  
  by
\begin{equation}\label{fractionalderivative.def1}
\widehat{i_\Omega f}(\xi)= \Omega(\xi) \hat f(\xi),\quad f \in {\mathcal S}_\infty.
\end{equation}
Clearly $i_\Omega$ is a continuous linear operator on  the closed linear  subspace ${\mathcal S}_\infty$
of ${\mathcal S}$.
For any function $f\in {\mathcal S}_\infty$, applying the   integration-by-parts technique $k_0$ times and noticing that
$\lim_{\epsilon\to 0}
\epsilon^{-\gamma} |\partial^{\bf i} \hat f(\epsilon \xi')|=0$ for all $\xi'\in S^{d-1}$ and ${\bf i}\in \ZZ_+^d$,
we obtain that
\begin{eqnarray}\label{extension.eq}
J_\Omega f({\bf x}) & = & \frac{(2\pi)^{-d}\Gamma(d-\gamma)}{ \Gamma(d+k_0-\gamma)}
\lim_{\epsilon\to 0}\int_{S^{d-1}}  \int_\epsilon^\infty r^{k_0+d-\gamma-1}\Omega(\xi')\nonumber  \\
& &
\quad \times   \Big(-\frac{d}{dr}\Big)^{k_0} \Big(e^{ir\langle {\bf x}, \xi'\rangle} \hat f(r\xi')\Big) dr d\sigma(\xi')\nonumber\\
& = & (2\pi)^{-d} \lim_{\epsilon\to 0}
\int_{S^{d-1}}  \int_\epsilon^\infty \Omega(\xi')r^{d-\gamma-1}  e^{ir\langle {\bf x}, \xi'\rangle} \hat f(r\xi') dr d\sigma(\xi')\nonumber\\
& = & (2\pi)^{-d}\int_{\RR^d} e^{i\langle {\bf x}, \xi\rangle} \Omega(\xi) \hat f(\xi) d\xi= i_\Omega f({\bf x}).
\end{eqnarray}
Hence the generalized Riesz potential $J_\Omega$ is the extension of the linear operator $i_\Omega$ from the closed
 subspace ${\mathcal S}_\infty$ to the whole space ${\mathcal S}$.

In the  sequel, we  will study  further properties of the generalized Riesz
potential  $J_\Omega$, such as the polynomial decay property
 (Theorem \ref{generalizedrieszomega1.tm}),
the continuity as a linear operator  from ${\mathcal S}$ to ${\mathcal S}'$ (Corollary \ref{generalizedrieszomega1.cr}),
the translation-invariance and  dilation-invariance (Theorem \ref{maintheorem.iomega1}),
the composition and left-inverse property (Theorem \ref{composition.tm} and  Corollary \ref{composition.cr}),
the uniqueness of various extensions
of the linear operator $i_\Omega$ from the closed
 subspace ${\mathcal S}_\infty$ to the whole space ${\mathcal S}$ (Theorems \ref{iomega2.tm} 
 and \ref{iomega4.tm}),
 the non-integrability in  the spatial domain (Theorem \ref{time1.tm}), and the non-integrability in  the Fourier domain (Theorem \ref{frequency.tm1}).
Some of  those properties  will be used to prove Theorem \ref{generalizedriesz.tm}, which  is included at
the end of this section.

\subsection{Polynomial decay  property  and continuity}

   \begin{Tm} \label{generalizedrieszomega1.tm} Let  $\gamma$ be a positive number with $\gamma-d\not\in \ZZ_+$,
  $k_0$ be the smallest nonnegative integer larger than $\gamma-d$,
 and let $\Omega\in C^\infty (\RR^d\backslash \{\bf 0\})$ be a  homogeneous function of degree $-\gamma$.
 If there exist positive constants $\epsilon$ and $C_\epsilon$ such that
\begin{equation}\label{generalizedrieszomega1.tm.eq1}
|f({\bf x})|\le C_\epsilon (1+\|{\bf x}\|)^{-k_0-d-\epsilon} \ {\rm for \ all} \ {\bf x}\in {\mathbb R}^d,
 \end{equation}
then there exists a positive constant $C$  such that
 \begin{eqnarray}\label{generalizedrieszomega1.tm.eq2}
|J_\Omega f({\bf x})| \le
  C  \Big(\sup_{{\bf z}\in {\mathbb R}^d} |f({\bf z})| (1+\|{\bf z}\|)^{k_0+d+\epsilon}\Big) (1+\|{\bf x}\|)^{\gamma-d},
 \ \  {\bf x}\in {\mathbb R}^d.
 \end{eqnarray}
 \end{Tm}

\begin{proof} 
Noting that
$  \big(\frac{d}{dr}\big)^{s} e^{ir\langle {\bf x}, \xi'\rangle}  
  =
 s! \Big(\sum_{|{\bf i}|=s} \frac{(i{\bf x})^{\bf i} \xi'^{\bf i}}{ {\bf i}!} \Big)e^{ir\langle {\bf x}, \xi'\rangle}$
 and
$\big(\frac{d}{dr}\big)^{k_0-s} \hat f(r\xi')= (k_0-s)! \sum_{|{\bf j}|=k_0-s}  \frac{ (\xi')^{\bf j} \partial^{\bf j} \hat f(r\xi')}{{\bf j}!}
$ for all $0\le s\le k_0$, we obtain from the Leibniz rule that
 \begin{eqnarray*}
\Big(\frac{d}{dr}\Big)^{k_0} \Big(e^{ir\langle {\bf x}, \xi'\rangle} \hat f(r\xi')\Big)
 & = &   \sum_{s=0}^{k_0}\binom{k_0}{s}
 \Big\{\Big(\frac{d}{dr}\Big)^{k_0-s} e^{ir\langle {\bf x}, \xi'\rangle}\Big\}\cdot
\Big\{\Big(\frac{d}{dr}\Big)^{k_0} \hat f(r\xi')\Big\}\nonumber\\
 & = &
 \Big(\sum_{|{\bf i}|+|{\bf j}|=k_0}\frac{k_0!}{{\bf i}!{\bf j}!}
{ (i{\bf x})}^{\bf i} (\xi')^{{\bf i}+{\bf j}}
\partial^{\bf j} \hat f(r\xi')\Big)e^{ir\langle {\bf x}, \xi'\rangle}.
 \end{eqnarray*}
 Substituting the above expression into \eqref{fractionalderivative.def}
 we get 
 \begin{eqnarray}\label{generalizedrieszomega1.tm.pf.eq1}
 J_\Omega f({\bf x})
& = & (-1)^{k_0}\sum_{|{\bf i}|+|{\bf j}|=k_0}\frac{k_0!}{{\bf i}!{\bf j}!} (i{\bf x})^{\bf i}
\Big\{\frac{(2\pi)^{-d}\Gamma(d-\gamma)}{ \Gamma(d+k_0-\gamma)}\nonumber\\
& & \times
\int_{\RR^d} e^{i\langle {\bf x}, \xi\rangle} \big(\xi^{{\bf i}+{\bf j}}\Omega(\xi)\big)
\partial^{\bf j}\hat  f(\xi)d\xi\Big\}\nonumber\\
& = & \frac{\Gamma(d-\gamma)}{\Gamma(d+k_0-\gamma)} \sum_{|{\bf i}|+|{\bf j}|=k_0}\frac{k_0! }{{\bf i}!{\bf j}!}  (-{\bf x})^{\bf i} J_{\Omega_{{\bf i}+{\bf j}}} (f_{\bf j})({\bf x}),
 \end{eqnarray}
 where $\Omega_{{\bf i}+{\bf j}}(\xi)= (i\xi)^{{\bf i}+{\bf j}}\Omega(\xi)$
  and $f_{\bf j}({\bf x})= {\bf x}^{\bf j} f({\bf x})$. Denote the inverse Fourier transform
  of $\Omega_{{\bf k}}, |{\bf k}|=k_0$, by $K_{{\bf k}}$. Then $K_{{\bf k}}\in C^\infty({\mathbb R}^d\backslash \{{\bf 0}\})$ is a homogeneous function of degree $\gamma-k_0-d$ (\cite[Theorems 7.1.16 and 7.1.18]{hormanderbook}), and hence
  there exists a positive constant $C$ such that
  \begin{equation}\label{generalizedrieszomega1.tm.pf.eq2}
  |K_{\bf k}({\bf x})|\le C \|{\bf x}\|^{\gamma-k_0-d} \quad {\rm for \ all} \  {\bf x}\in {\mathbb R}^d\backslash \{{\bf 0}\}.
  \end{equation}
 For any $\epsilon>0$ and $\beta\in (0, d)$, we have
  \begin{eqnarray}\label{generalizedrieszomega1.tm.pf.eq3}
  & &  \int_{{\mathbb R}^d} \|{\bf x}-{\bf y}\|^{-\beta} (1+\|{\bf y}\|)^{-d-\epsilon} d{\bf y}\nonumber\\
  & \le &
  \Big( \int_{\|{\bf y}\|\le (\|{\bf x}\|+1)/2}+\int_{(\|{\bf x}\|+1)/2\le \|{\bf y}\|\le 2(\|{\bf x}\|+1)}+
  \int_{\|{\bf y}\|\ge 2(\|{\bf x}\|+1)} \Big)  \nonumber\\
  & & \qquad  \|{\bf x}-{\bf y}\|^{-\beta} (1+\|{\bf y}\|)^{-d-\epsilon} d{\bf y}\nonumber\\
& \le &  C (1+\|{\bf x}\|)^{-\beta}.
  \end{eqnarray}
  Combining \eqref{generalizedrieszomega1.tm.pf.eq1},
  \eqref{generalizedrieszomega1.tm.pf.eq2} and \eqref{generalizedrieszomega1.tm.pf.eq3} yields
  \begin{eqnarray*}
  |J_\Omega f({\bf x})| & \le &  C \sum_{|{\bf i}|+|{\bf j}|=k_0}
  |{\bf x}|^{|{\bf i}|} \Big|\int_{{\mathbb R}^d} K_{{\bf i}+{\bf j}}({\bf x}-{\bf y}) {\bf y}^{\bf j} f({\bf y}) \Big| d{\bf y}\nonumber\\
  & \le & C (1+\|{\bf x}\|)^{k_0} \int_{{\mathbb R}^d}
   \|{\bf x}-{\bf y}\|^{\gamma-k_0-d}
  (1+\|{\bf y}\|)^{k_0} |f({\bf y})| d{\bf y}\nonumber\\
& \le & C  \Big(\sup_{{\bf z}\in {\mathbb R}^d} |f({\bf z})| (1+\|{\bf z}\|)^{k_0+d+\epsilon}\Big) (1+\|{\bf x}\|)^{\gamma-d}.\end{eqnarray*}
This proves the desired polynomial decay estimate \eqref{generalizedrieszomega1.tm.eq2}.
\end{proof}

For any  $f\in {\mathcal S}$ and ${\bf j}\in \ZZ_+^d$ with $|{\bf j}|=1$, it follows from \eqref{fractionalderivative.def} that
  \begin{eqnarray*}
\partial^{\bf j}(J_\Omega f) ({\bf x}) &= &  J_\Omega(\partial^{\bf j} f) ({\bf x})\nonumber\\
  & = &
  \frac{(2\pi)^{-d}\Gamma(d-\gamma)}{ \Gamma(d+k_0-\gamma)} \int_{S^{d-1}}  \int_0^\infty \Omega(\xi') (i\xi')^{\bf j} r^{k_0+d-\gamma-1}\nonumber  \\
& &
\times   \Big(-\frac{d}{dr}\Big)^{k_0} \Big(e^{ir\langle {\bf x}, \xi'\rangle} \hat f(r\xi') r \Big) dr d\sigma(\xi')\nonumber\\
& = &
  \frac{(2\pi)^{-d}\Gamma(d-\gamma)}{ \Gamma(d+k_0-\gamma)} \int_{S^{d-1}}  \int_0^\infty \Omega(\xi') (i\xi')^{\bf j} r^{k_0+d-\gamma-1}\nonumber  \\
& &
\times   \Big\{ r \Big(-\frac{d}{dr}\Big)^{k_0} \Big(e^{ir\langle {\bf x}, \xi'\rangle} \hat f(r\xi')  \Big)
 \nonumber\\
 & & \quad  -{k_0}\Big(-\frac{d}{dr}\Big)^{k_0-1} \Big(e^{ir\langle {\bf x}, \xi'\rangle} \hat f(r\xi')  \Big)\Big\}  dr d\sigma(\xi')\nonumber\\
 & = & \Big(\frac{d+k_0-\gamma}{d-\gamma}-k_0\frac{1}{d-\gamma}\Big) J_{\Omega_{\bf j}} f({\bf x})=
 J_{\Omega_{\bf j}} f({\bf x}),
  \end{eqnarray*}
  where $\Omega_{\bf j}(\xi)=(i\xi)^{\bf j} \Omega(\xi) $.
 Applying the  argument inductively leads to
   \begin{equation} \label{fractionalderivative.eq00}
\partial^{\bf j} (J_\Omega f)=  J_\Omega(\partial^{\bf j} f) = J_{\Omega_{\bf j}} f \quad {\rm for \ all} \ f\in {\mathcal S} \ {\rm and} \ {\bf j}\in \ZZ_+^d,
  \end{equation}
  where $\Omega_{\bf j}(\xi)= (i\xi)^{\bf j}\Omega(\xi)$. This together with Theorem \ref{generalizedrieszomega1.tm} shows that
$J_\Omega f$ is a smooth function on ${\mathbb R}^d$ for any Schwartz function $f$.

   \begin{Cr} \label{generalizedrieszomega1.cr0} Let  $\gamma, k_0$ and $\Omega$ be as
   in Theorem \ref{generalizedrieszomega1.tm}.
 If  $f$ satisfies \eqref{generalizedrieszomega1.tm.eq1} for some  positive constants $\epsilon$ and $C_\epsilon$,
then  for any ${\bf j}\in \ZZ_+^d$ with $|{\bf j}|<\gamma$ there exists a positive constant $C_{\bf j}$  such that
 \begin{equation}\label{generalizedrieszomega1.cr0.eq2}
|\partial^{\bf j} (J_\Omega f)({\bf x})| \le
  C_{\bf j}  \Big(\sup_{{\bf z}\in {\mathbb R}^d} |f({\bf z})| (1+\|{\bf z}\|)^{k_0+d+\epsilon}\Big) (1+\|{\bf x}\|)^{\gamma-|{\bf j}|-d}, \ {\bf x}\in \RR^d.
  \end{equation}
 \end{Cr}

 An easy consequence of the above smoothness result about  $J_\Omega f$   is the
  continuity of  the generalized Riesz potential $J_\Omega$ from ${\mathcal S}$ to ${\mathcal S}'$.

\begin{Cr}\label{generalizedrieszomega1.cr}
 Let  $\gamma$ be a positive number with $\gamma-d\not\in \ZZ_+$,
   and let $\Omega\in C^{\infty} (\RR^d\backslash \{\bf 0\})$ be a  homogeneous function of degree $-\gamma$.
 Then the generalized Riesz potential  $J_\Omega $ associated with the homogeneous function $\Omega$ is
   a  continuous linear operator from ${\mathcal S}$ to ${\mathcal S}'$.
\end{Cr}

Now consider the generalized Riesz potential $J_\Omega$ when $\Omega$ is a homogeneous function  of positive degree $\alpha$.
In this case,
$$ J_\Omega f({\bf x})  =   (2\pi)^{-d} \int_{\RR^d} e^{i\langle {\bf x}, \xi\rangle}
\Omega(\xi)\hat f(\xi) d\xi\quad {\rm for \ all} \ f\in {\mathcal S}$$
by \eqref{fractionalderivative.neweq2}.
Applying the integration-by-parts technique then gives
$$
 J_\Omega f({\bf x})
=  (2\pi)^{-d} (-i {\bf x}^{\bf i})^{-1} \sum_{{\bf j}+{\bf k}={\bf i}}
\frac{{\bf i}!}{{\bf j}!{\bf k}!}
\int_{\RR^d}
e^{i\langle {\bf x}, \xi\rangle}
\partial^{\bf j} \Omega(\xi) \partial^{\bf k}\hat f(\xi) d\xi
$$
for any ${\bf i}\in \ZZ_+^d$.
This, together with the identity
$$1=\sum_{|{\bf l}|=\lceil \alpha\rceil-|{\bf j}|} \frac{(\lceil \alpha \rceil-|{\bf j}|)!}{{\bf l}!}
\Big(\frac{i\xi}{\|{\bf \xi}\|^2}\Big)^{\bf l}
(-i\xi)^{\bf l},\quad |{\bf j}|\le \lceil \alpha\rceil,$$
leads to the following estimate of $J_\Omega f({\bf x})$:
\begin{eqnarray*}
|J_\Omega f({\bf x})|\!\! & \le & \!\! C (1+\|{\bf x}\|)^{-\lceil \alpha\rceil}
\sum_{|{\bf j}|+|{\bf k}|\le \lceil \alpha\rceil, |{\bf l}|=\lceil \alpha\rceil-|{\bf j}|}
\Big|\int_{{\RR}^d} e^{i\langle {\bf x}, \xi\rangle}
 \Omega_{{\bf j}, {\bf l}}(\xi) \xi^{\bf l} \partial^{\bf k} \hat f(\xi) d\xi\Big| \nonumber\\
\!\! & \le & \!\!
C (1+\|{\bf x}\|)^{-\lceil \alpha\rceil}
\sum_{|{\bf j}|+|{\bf k}|\le \lceil \alpha\rceil, |{\bf l}|+|{\bf j}|=\lceil \alpha\rceil} |I_{\Omega_{{\bf j}, {\bf l}}} f_{{\bf l}, {\bf k}}({\bf x})|,
\end{eqnarray*}
where $\lceil \alpha\rceil$ is the smallest integer larger than $\alpha$,  $\Omega_{{\bf j}, {\bf l}}(\xi)=\partial^{\bf j}\Omega(\xi) (i\xi/\|\xi\|^2)^{\bf l}$, and
 $\widehat{ f_{{\bf l}, {\bf k}}} (\xi)=(-i\xi)^{\bf l} \partial^{{\bf k}} \hat f(\xi)$.
Note that $\Omega_{{\bf j}, {\bf l}}\in C^\infty(\RR^d\backslash \{{\bf 0}\})$ is a homogeneous function of degree $\alpha-\lceil \alpha\rceil<0$
when $|{\bf j}|+|{\bf l}|=\lceil \alpha\rceil$, and also that functions
$f_{{\bf l}, {\bf k}}({\bf x}), |{\bf k}|, |{\bf l}|\le \lceil \alpha\rceil$  are linear combinations of
${\bf x}^{\bf i} \partial^{\bf j} f({\bf x}), |{\bf i}|, |{\bf j}|\le  \lceil \alpha\rceil$.
We then apply Theorem \ref{generalizedrieszomega1.tm}  to obtain the following polynomial  decay estimate of $J_\Omega f$ when $\Omega$ is
a  homogeneous function of positive degree:

\begin{pr}\label{positivegeneralizedrieszomega1.cr}
Let $\alpha$ be a positive non-integer number, and
 $\Omega\in C^\infty (\RR^d\backslash \{\bf 0\})$ be a  homogeneous function of degree $\alpha$.
 If there exist positive constants $\epsilon$ and $C_\epsilon$ such that
\begin{equation*}\label{positivegeneralizedrieszomega1.cr.eq1}
\sum_{|{\bf i}|\le \lceil \alpha\rceil}
|\partial^{\bf i} f({\bf x})|\le C_\epsilon (1+\|{\bf x}\|)^{-\lceil \alpha\rceil-d-\epsilon} \ {\rm for \ all} \ {\bf x}\in {\mathbb R}^d,
 \end{equation*}
then there exists a positive constant $C$  such that
 \begin{equation}
 \label{positivegeneralizedrieszomega1.cr.eq2}
|J_\Omega f({\bf x})| \le
  C  \Big(\sum_{|{\bf i}|\le \lceil \alpha\rceil} \sup_{{\bf z}\in {\mathbb R}^d} |\partial^{\bf i} f({\bf z})| (1+\|{\bf z}\|)^{\lceil \alpha\rceil+d+\epsilon}\Big) (1+\|{\bf x}\|)^{-\alpha-d}
 \end{equation}
 for all ${\bf x}\in {\mathbb R}^d$.
\end{pr}

The estimates in   \eqref{generalizedrieszomega1.tm.eq2} and
\eqref{positivegeneralizedrieszomega1.cr.eq2} indicate that  the generalized Riesz potential $J_\Omega f$ has faster polynomial
 decay at infinity when the degree of the homogeneous function $\Omega$ becomes larger.
Next, we show that the generalized Riesz potential $J_\Omega f$ has  faster polynomial decay at infinity
when $f$ has  vanishing moments up to some order; i.e.,
\begin{equation}\label{momentcondition}
\int_{\RR^d} {\bf x}^{\bf i} f({\bf x}) d{\bf x}=0,\ |{\bf i}|\le m_0
\end{equation}
where $m_0\ge 0$.
In this case, $\partial^{\bf i}\hat f({\bf 0})=0$ for all $|{\bf i}|\le m_0$, and hence
\begin{equation}
\hat f(\xi)= \sum_{|{\bf k}|=m_0+1}\frac{m_0+1}{{\bf k}!} \int_0^1 \xi^{\bf k} \partial^{\bf k} \hat f(t\xi) (1-t)^{m_0} dt
\end{equation}
by the Taylor expansion to $\hat f$ at the origin.
Now we assume that $\Omega\in C^\infty(\RR^d\backslash\{{\bf 0}\})$ is a homogeneous function of degree
$\alpha\in (-m_0-1, \infty)\backslash \ZZ $.
Then
\begin{eqnarray}\label{momentgeneralizedrieszomega1.pr.pf1}
|J_\Omega f({\bf x})|\!\! \!& \le &\!\!  
C \sum_{|{\bf k}|=m_0+1}\int_0^1\int_{\|\xi\|\le 1} |\xi|^{\alpha+m_0+1} |\partial^{\bf k} \hat f(t\xi)| d\xi dt+
C \int_{|\xi|\ge 1} |\xi|^{\alpha} |\hat f(\xi)|d\xi\nonumber\\
& \le &  C \sum_{|{\bf i}|\le m_0+1} \sup_{\xi \in \RR^d}
\big((1+\|\xi\|)^{\lceil \alpha\rceil +d} |\partial^{\bf i}\hat f(\xi)|\big) 
 \end{eqnarray}
for all ${\bf x}\in \RR^d$ with $\|{\bf x}\|\le 1$, and
 \begin{eqnarray}\label{momentgeneralizedrieszomega1.pr.pf2}
|J_\Omega f({\bf x})|\!\! &\le &\!\! C \sum_{|{\bf k}|=m_0+1}
\int_0^1
\Big|\int_{\RR^d} e^{-i\langle {\bf x}, \xi\rangle} \phi(\|{\bf x}\| \xi) \xi^{\bf k} \Omega(\xi)  \partial^{\bf k} \hat f(t\xi)  d\xi\Big| dt\nonumber\\
\!\! & & \!\!
+C\sum_{|{\bf k}|=m_0+1}
\int_0^1  \Big|\int_{\RR^d} e^{-i\langle {\bf x}, \xi\rangle} \big(\phi(\xi)-\phi(\|{\bf x}\| \xi)\big)\xi^{\bf k} \Omega(\xi)
\partial^{\bf k} \hat f(t\xi)  d\xi\Big|dt\nonumber
\\
\!\!& &\!\! + C
\Big|\int_{\RR^d} e^{-i\langle {\bf x}, \xi\rangle} \big(1-\phi(\xi)\big) \Omega(\xi) \hat f(\xi)  d\xi\Big|\nonumber\\
\!\!&\le & \!\!
C (1+\|{\bf x}\|)^{-\lceil \alpha\rceil-m_0-d}
\Big\{\sum_{|{\bf k}|=m_0+1, |{\bf j}|\le \lceil \alpha\rceil+m_0+d}\nonumber\\
\!\!& &\!\!
\quad \int_0^1\int_{\RR^d} \Big|\partial^{\bf j} \big(\phi(\|{\bf x}\|\xi) \xi^{\bf k} \Omega(\xi) \partial^{\bf k}\hat f(t\xi)\big)\Big| d\xi dt\Big\}
\nonumber\\
\!\!& &\!\!
+
C (1+\|{\bf x}\|)^{-\lceil \alpha\rceil-m_0-d-1} \Big\{\sum_{|{\bf k}|=m_0+1, |{\bf j}|\le \lceil \alpha\rceil+m_0+d+1}\nonumber\\
\!\!& &\!\! \quad \int_0^1\int_{\RR^d} \Big|\partial^{\bf j}\big( (\phi(\xi)-\phi(\|{\bf x}\|\xi)) \xi^{\bf k} \Omega(\xi) \partial^{\bf k}\hat f(t\xi)\big)\Big| d\xi dt\Big\}\nonumber\\
\!\!& &\!\! +
C (1+\|{\bf x}\|)^{-\lceil \alpha\rceil-m_0-d-1} \nonumber\\
& & \quad \times \Big\{\sum_{|{\bf j}|\le \lceil \alpha\rceil+m_0+d+1}
\int_{\RR^d} \Big|\partial^{\bf j}\big( (1-\phi(\xi))  \Omega(\xi) \hat f(\xi)\big)\Big| d\xi\Big\} \nonumber\\
\!\!& \le & \!\!
C \Big(\sum_{|{\bf i}|\le \lceil \alpha\rceil+2m_0+d+2} \sup_{\xi\in \RR^d} (1+\|\xi\|)^{\lceil \alpha\rceil+d} |\partial^{\bf i} \hat f(\xi)|\Big) (1+\|{\bf x}\|)^{-\alpha-m_0-d-1} 
\end{eqnarray}
for all ${\bf x}\in \RR^d$ with $\|{\bf x}\|\ge 1$, where  $\phi$ is a  $C^\infty$ function such that
   $\phi(\xi)=1$ for all $\xi$ in the unit ball $B({\bf 0}, 1)$ centered at the origin, and
   $\phi(\xi)=0$ for all $\xi$ not in the ball $B({\bf 0}, 2)$ with radius 2 and center at the origin.
   This proves the following result about the generalized Riesz potential $J_\Omega f$ when $f$ has vanishing moments upto some order.

   \begin{pr}\label{momentgeneralizedrieszomega1.pr}
   Let $m_0\ge 0, \alpha\in (-m_0-1, \infty)\backslash \ZZ$, and $\Omega\in C^\infty(\RR^d\backslash\{{\bf 0}\})$ be a homogeneous function of degree $\alpha$. Then the following statements hold.

   \begin{itemize}

   \item [{(i)}] If $f$ satisfies \eqref{momentcondition} and
   \begin{equation}
   \sum_{|{\bf i}|\le\lceil \alpha\rceil+2m_0+d+2} \sup_{\xi\in \RR^d} (1+\|\xi\|)^{\lceil \alpha\rceil+d} |\partial^{\bf i} \hat f(\xi)|<\infty,
   \end{equation}
   then there exists a positive constant $C$ such that
   \begin{eqnarray}
   |J_\Omega f({\bf x})| & \le &  C  \Big(\sum_{|{\bf i}|\le\lceil \alpha\rceil+2m_0+d+2} \sup_{\xi\in \RR^d} (1+\|\xi\|)^{\lceil \alpha\rceil+d} |\partial^{\bf i} \hat f(\xi)|\Big)
\nonumber\\
& & \quad \times (1+\|{\bf x}\|)^{-\alpha-m_0-d-1} \quad {\rm for \ all} \ {\bf x}\in \RR^d.
   \end{eqnarray}

   \item[{(ii)}] If $f$ satisfies \eqref{momentcondition} and
   \begin{equation}
  \sum_{|{\bf i}|\le \max(\lceil \alpha\rceil +d, 0)}
\sup_{{\bf z}\in \RR^d} \big( (1+\|{\bf z}\|)^{\lceil \alpha\rceil+2m_0+2d+2+\epsilon} |\partial^{\bf i} f({\bf z})|\big)<\infty
   \end{equation}
 for some $\epsilon>0$,  then
      \begin{eqnarray}
   |J_\Omega f({\bf x})| & \le &  C  \Big(    \sum_{|{\bf i}|\le \max(\lceil \alpha\rceil +d, 0)}
\sup_{{\bf z}\in \RR^d} \big( (1+\|{\bf z}\|)^{\lceil \alpha\rceil+2m_0+2d+2+\epsilon} |\partial^{\bf i} f({\bf z})|
\Big)
\nonumber\\
& & \quad \times (1+\|{\bf x}\|)^{-\alpha-m_0-d-1} \quad {\rm for \ all} \ {\bf x}\in \RR^d.
   \end{eqnarray}
\end{itemize}
      \end{pr}

The conclusions in Proposition \ref{momentgeneralizedrieszomega1.pr} do not apply to the generalized Riesz potential $J_\Omega f$
where $\Omega\in C^\infty(\RR^d\backslash \{{\bf 0}\})$  is a homogeneous function of degree zero.
In this case, applying the argument used to establish
\eqref{momentgeneralizedrieszomega1.pr.pf1} and \eqref{momentgeneralizedrieszomega1.pr.pf2}, we have
that
\begin{eqnarray}\label{momentgeneralizedrieszomega2.pr.pf1}
|J_\Omega f({\bf x})|
& \le &  C \sum_{|{\bf i}|\le m_0+1} \sup_{\xi \in \RR^d}
\big((1+\|\xi\|)^{d+\epsilon} |\partial^{\bf i}\hat f(\xi)|\big) 
 \end{eqnarray}
for all ${\bf x}\in \RR^d$ with $\|{\bf x}\|\le 1$, and
\begin{eqnarray}\label{momentgeneralizedrieszomega2.pr.pf2}
|J_\Omega f({\bf x})|
\!\!&\le & \!\!
C (1+\|{\bf x}\|)^{-m_0-d} \nonumber\\
& & \times \Big\{
\sum_{|{\bf k}|=m_0+1, |{\bf j}|\le m_0+d} \int_0^1\int_{\RR^d} \big|\partial^{\bf j} \big(\phi(\|{\bf x}\|\xi) \xi^{\bf k} \Omega(\xi) \partial^{\bf k}\hat f(t\xi)\big)\big| d\xi dt\Big\}
\nonumber\\
\!\!& &\!\!
+
C (1+\|{\bf x}\|)^{-m_0-d-1} \sum_{|{\bf k}|=m_0+1,  |{\bf j}|+|{\bf l}|\le m_0+d+1, |{\bf j}|\le m_0+d}\nonumber\\
\!\!& &\!\! \quad \int_0^1\int_{\RR^d} \big|\partial^{\bf j}\big( (\phi(\xi)-\phi(\|{\bf x}\|\xi)) \xi^{\bf k} \Omega(\xi)\big)\big|\times
 \big| \partial^{{\bf k}+{\bf l}}\hat f(t\xi)\big)\big| d\xi dt\Big\}\nonumber\\
\!\!& &\!\!
+
C (1+\|{\bf x}\|)^{-m_0-d-2} \sum_{|{\bf k}|=m_0+1, |{\bf j}|+|{\bf l}|\le m_0+d+2, |{\bf l}|\le 1}\nonumber\\
\!\!& &\!\! \quad \int_0^1\int_{\RR^d} \big|\partial^{\bf j}\big( (\phi(\xi)-\phi(\|{\bf x}\|\xi)) \xi^{\bf k} \Omega(\xi)\big)\big|\times
 \big| \partial^{{\bf k}+{\bf l}}\hat f(t\xi)\big)\big| d\xi dt\Big\}\nonumber\\
 \!\!& &\!\! +
C (1+\|{\bf x}\|)^{-m_0-d-1} 
\sum_{|{\bf j}|\le m_0+d+1}
\int_{\RR^d} \big|\partial^{\bf j}\big( (1-\phi(\xi))  \Omega(\xi) \hat f(\xi)\big)\big| d\xi \nonumber\\
\!\!& \le & \!\!
C \Big(\sum_{|{\bf i}|\le 2m_0+d+2} \sup_{\xi\in \RR^d} (1+\|\xi\|)^{d+\epsilon} |\partial^{\bf i} \hat f(\xi)|\Big) (1+\|{\bf x}\|)^{-m_0-d-1}
\end{eqnarray}
for all ${\bf x}\in \RR^d$ with $\|{\bf x}\|\ge 1$, where $\epsilon\in (0,1)$.
Therefore

  \begin{pr}\label{momentgeneralizedrieszomega2.pr}
   Let  $\Omega\in C^\infty(\RR^d\backslash\{{\bf 0}\})$ be a homogeneous function of degree zero. Then the following statements hold.

   \begin{itemize}

   \item [{(i)}] If $f$ satisfies \eqref{momentcondition} for some $m_0\ge 0$ and
 $$   \sum_{|{\bf i}|\le 2m_0+d+2} \sup_{\xi\in \RR^d} (1+\|\xi\|)^{d+\epsilon} |\partial^{\bf i} \hat f(\xi)|<\infty$$ for some $\epsilon>0,$
   then there exists a positive constant $C$ such that
 $$
   |J_\Omega f({\bf x})|  \le    C  \Big(\sum_{|{\bf i}|\le 2m_0+d+2} \sup_{\xi\in \RR^d} (1+\|\xi\|)^{d+\epsilon} |\partial^{\bf i} \hat f(\xi)|\Big)
 (1+\|{\bf x}\|)^{-m_0-d-1}
 \quad {\rm for \ all} \ {\bf x}\in \RR^d.
 $$

   \item[{(ii)}] If $f$ satisfies \eqref{momentcondition} for some $m_0\ge 0$ and
  $$
  \sum_{|{\bf i}|\le d+1}
\sup_{{\bf z}\in \RR^d} \big( (1+\|{\bf z}\|)^{2m_0+2d+2+\epsilon} |\partial^{\bf i} f({\bf z})|\big)<\infty
$$
 for some $\epsilon>0$,  then
      \begin{eqnarray*}
   |J_\Omega f({\bf x})| & \le &  C  \Big(    \sum_{|{\bf i}|\le d+1}
\sup_{{\bf z}\in \RR^d} \big( (1+\|{\bf z}\|)^{2m_0+2d+2+\epsilon} |\partial^{\bf i} f({\bf z})|
\Big)
\nonumber\\
& & \quad \times (1+\|{\bf x}\|)^{-m_0-d-1} \quad {\rm for \ all} \ {\bf x}\in \RR^d.
   \end{eqnarray*}
\end{itemize}
      \end{pr}

%

\subsection{Translation-invariance and dilation-invariance}
In this subsection, we show that
the generalized Riesz potential  $J_\Omega$ from ${\mathcal S}$ to ${\mathcal S}'$
is dilation-invariant and translation-invariant,   and that its restriction on the closed subspace ${\mathcal S}_\infty$
of ${\mathcal S}$ is the same as the linear
operator $i_\Omega$ on  ${\mathcal S}_\infty$.

\begin{Tm}\label{maintheorem.iomega1}
Let $\gamma\in \RR$  with $\gamma-d\not \in \ZZ_+$,
 $\Omega\in C^\infty(\RR^d\backslash \{\bf 0\})$ be a homogeneous function of degree $-\gamma$,
 and let  $J_\Omega$  be defined by \eqref{fractionalderivative.def}. Then
  \begin{itemize}
   \item[{(i)}] $J_\Omega$ is dilation-invariant;
 \item[{(ii)}] $J_\Omega$ is translation-invariant; and


\item[{(iii)}] $\widehat{J_\Omega f}(\xi)= \Omega(\xi) \hat f(\xi)$ for any function $f\in {\mathcal S}_\infty$.
 \end{itemize}
\end{Tm}

\begin{proof} {\em (i)}\quad
 For any  $f\in {\mathcal S}$ and any $t>0$,  
 \begin{eqnarray*}
 J_\Omega (\delta_t f) ({\bf x}) & = & \frac{(2\pi t)^{-d}\Gamma(d-\gamma)}{ \Gamma(d+k_0-\gamma)}
 \int_{S^{d-1}}  \int_0^\infty \Omega(\xi')r^{k_0-\gamma+d-1}\nonumber  \\
& &
\times   \Big(-\frac{d}{dr}\Big)^{k_0} \Big(e^{ir\langle {\bf x}, \xi'\rangle} \hat f(r\xi'/t)\Big) dr d\sigma(\xi')
= t^{-\gamma} \delta_t(J_\Omega f)({\bf x}),
 \end{eqnarray*}
 where the first equality follows from $\widehat {\delta_t f}(\xi)=t^{-d} \hat f(\xi/t)$
   and the second equality is obtained by change of variables.
  This leads to  the dilation-invariance of the generalized Riesz potential $J_\Omega$.

{\em (ii)}\quad
For any  $f\in {\mathcal S}$ and a vector ${\bf x}_0\in \RR^d$,
 we obtain  from \eqref{fractionalderivative.def} that
 \begin{eqnarray*}
 J_\Omega (\tau_{{\bf x}_0}f)({\bf x})
     & = &   \frac{(2\pi)^{-d}\Gamma(d-\gamma)}{ \Gamma(d+k_0-\gamma)} \int_{S^{d-1}}  \int_0^\infty r^{k_0-\gamma+d-1}   \Omega(\xi') \nonumber  \\
& & \times
\big(-\frac{d}{dr}\big)^{k_0} \Big(e^{ir\langle {\bf x}-{\bf x}_0, \xi'\rangle} \hat f(r\xi')\Big) dr d\sigma(\xi')=   J_\Omega f( {\bf x}-{\bf x}_0),
 \end{eqnarray*}
where $k_0$ is a nonnegative integer larger than $\gamma-d$.  This shows that
 the generalized Riesz potential $J_\Omega$
 is  translation-invariant.


{\em (iii)}\quad
The third conclusion follows by taking Fourier transform of the equation \eqref{extension.eq} on both sides.
\end{proof}


\subsection{Composition  and left-inverse}
In this subsection, we consider the composition  and left-inverse properties of generalized Riesz potentials.

\begin{Tm}\label{composition.tm}
Let $\gamma_1$ and $\gamma_2\in \RR$ satisfy
 $\gamma_2<d, \gamma_1+\gamma_2<d$ and $\gamma_1-d\not\in \ZZ_+$,
and let $\Omega_1, \Omega_2\in C^\infty(\RR^d\backslash\{{\bf 0}\})$
be homogeneous  functions of degree $-\gamma_1$ and $-\gamma_2$ respectively.
Then
\begin{equation}\label{composition.tm.eq1}
J_{\Omega_1}(J_{\Omega_2} f)=J_{\Omega_1\Omega_2} f \quad {\rm for\ all} \ f\in {\mathcal S}.
\end{equation}
\end{Tm}

As a consequence of Theorem \ref{composition.tm}, we have the following result about left-invertibility
of the generalized Riesz potential $J_\Omega$.

 \begin{Cr}\label{composition.cr}
Let $\gamma\in (-d, \infty)$ with $\gamma-d\not\in \ZZ_+$ and
  $\Omega\in C^\infty(\RR^d\backslash\{{\bf 0}\})$
be homogeneous of degree $-\gamma$ with $\Omega(\xi)\ne 0$ for all $\xi\in S^{d-1}$.
Then $J_{\Omega}J_{\Omega^{-1}}$ is an identity operator on ${\mathcal S}$.
If we further assume that $\gamma\in (-d, d)$, then both
$J_{\Omega^{-1}}J_\Omega$ and $J_{\Omega}J_{\Omega^{-1}}$ are identity operators on ${\mathcal S}$.
\end{Cr}

Taking $\Omega(\xi)=\|\xi\|^{-\gamma}$ in the above corollary yields that the linear operator  $I_\gamma$ in \eqref{generalizedriesz.def} is a left-inverse of the fractional Laplacian $(-\triangle)^{\gamma/2}$.

\begin{Cr}\label{invariantleftinverse.cr}
Let $\gamma$ be a positive number with $\gamma-d\not\in \ZZ_+$. Then $I_\gamma$ is a left-inverse
of the fractional Laplacian $(-\triangle)^{\gamma/2}$.
\end{Cr}

\begin{proof}[Proof of Theorem \ref{composition.tm}]
   Let $k_0$ be the smallest nonnegative integer such that $k_0-\gamma_1+d>0$, and set $\Omega(\xi)=\Omega_1(\xi)\Omega_2(\xi)$.
  If $k_0=0$, then
the conclusion \eqref{composition.tm.eq1} follows from
\eqref{fractionalderivative.neweq2}. Now we assume that $k_0\ge 1$.
Then
\begin{eqnarray*}\label{composition.tm.pf.eq2}
 J_{\Omega_1}(J_{\Omega_2} f)({\bf x})
\!\! 
& = &
  \frac{(2\pi)^d\Gamma(d-\gamma_1)}{ \Gamma(d+k_0-\gamma_1)} \lim_{\epsilon\to 0} \int_{S^{d-1}}  \int_\epsilon^\infty \Omega(\xi')
   r^{k_0+d-\gamma_1-1}  \nonumber  \\
& &\times
\Big\{ r \Big(-\frac{d}{dr}\Big)^{k_0} \Big(e^{ir\langle {\bf x}, \xi'\rangle} \hat f(r\xi') r^{-\gamma_2-1} \Big)\nonumber\\
& & -k_0 \Big(-\frac{d}{dr}\Big)^{k_0-1} \Big(e^{ir\langle {\bf x}, \xi'\rangle} \hat f(r\xi') r^{-\gamma_2-1} \Big)\Big\}
  dr d\sigma(\xi')\nonumber\\
& = &  \frac{(2\pi)^d\Gamma(d+1-\gamma_1)}{ \Gamma(d+k_0-\gamma_1)} \lim_{\epsilon\to 0} \int_{S^{d-1}}  \int_\epsilon^\infty \Omega(\xi')
   r^{k_0+d-\gamma_1-1}  \nonumber  \\
& &\times
 \Big(-\frac{d}{dr}\Big)^{k_0-1} \Big(e^{ir\langle {\bf x}, \xi'\rangle} \hat f(r\xi') r^{-\gamma_2-1} \Big) dr d\sigma(\xi')\nonumber\\
 &  = & \cdots\nonumber\\
 & = &
   \frac{ (2\pi)^{-d}\Gamma(d+k_0-\gamma_1)}{ \Gamma(d+k_0-\gamma_1)} \lim_{\epsilon\to 0} \int_{S^{d-1}}  \int_\epsilon^\infty \Omega(\xi')
   r^{k_0+d-\gamma_1-1}  \nonumber  \\
& &\times
  \Big(e^{ir\langle {\bf x}, \xi'\rangle} \hat f(r\xi') r^{-\gamma_2-k_0} \Big) dr d\sigma(\xi')\nonumber\\
  & = & J_{\Omega_1\Omega_2} f({\bf x})\quad {\rm for\ all} \ {\bf x}\in \RR^d,
\end{eqnarray*}
where the second equality is obtained by  applying the integration-by-parts technique and using the fact that
$ \epsilon^{k_0+d-\gamma_1} \big(\frac{d}{dr}\big)^{k_0-1} \big(e^{ir\langle {\bf x}, \xi'\rangle} \hat f(r\xi') r^{-\gamma_2-1} \big)\big|_{r=\epsilon}$  converges to  zero uniformly on $\ \xi\in S^{d-1}$ under the assumption that $\gamma_1+\gamma_2<d$.
The conclusion
\eqref{composition.tm.eq1} then follows.
\end{proof}


\subsection{Translation-invariant  and dilation-invariant extensions of the linear operator $i_\Omega$}
In this subsection, we show that   the generalized Riesz potential $J_\Omega$ in \eqref{fractionalderivative.def} is the {\bf only}
  continuous linear operator from ${\mathcal S}$ to ${\mathcal S}'$ that is translation-invariant and dilation-invariant, and
 that is an extension of  the linear operator $i_\Omega$ in \eqref{fractionalderivative.def1} from the closed subspace ${\mathcal S}_\infty$ to the whole space ${\mathcal S}$.

  \begin{Tm}\label{iomega2.tm}
Let $\gamma$ be a positive  number with $\gamma-d\not\in \ZZ_+$,
 $\Omega\in C^\infty(\RR^d\backslash \{\bf 0\})$ be a nonzero homogeneous function of degree $-\gamma$,
 and let $J_\Omega$ be defined by \eqref{fractionalderivative.def}. Then
 $I$ is a continuous linear  operator from ${\mathcal S}$ to ${\mathcal S}'$  such that
$I$ is dilation-invariant and translation-invariant,  and   that the restriction of $I$
on ${\mathcal S}_\infty$ is
the  same as the linear operator $i_\Omega$ in
\eqref{fractionalderivative.def1} if and only if $I=J_\Omega$.
\end{Tm}

To prove Theorem \ref{iomega2.tm}, we need two technical lemmas about  extensions of the linear operator $i_\Omega$ on ${\mathcal S}_\infty$.

\begin{Lm}\label{homogeneous1.lm}
Let $\gamma$ be a positive number with $\gamma-d\not\in  \ZZ_+$,
 $\Omega\in C^\infty(\RR^d\backslash \{\bf 0\})$ be a homogeneous function of degree $-\gamma$,
 and let $J_\Omega$ be defined by \eqref{fractionalderivative.def}.
 Then
  a continuous linear  operator $I$ from ${\mathcal S}$ to ${\mathcal S}'$  is an extension of the linear operator $i_\Omega$  on ${\mathcal S}_\infty$ if and only if
\begin{equation}\label{homogeneous1.lm.eq1}
If= J_\Omega f+\sum_{|{\bf i}|\le N} \frac{\partial^{\bf i} \hat f({\bf 0})}{{\bf i}!} H_{\bf i}
\end{equation}
 for some integer $N$ and tempered distributions $H_{\bf i}, {\bf i}\in \ZZ_+^d$ with $|{\bf i}|\le N$.
 \end{Lm}

\begin{proof} The sufficiency  follows from Theorem \ref{maintheorem.iomega1} and the assumption that $H_{\bf i}, |{\bf i}|\le N$, in \eqref{homogeneous1.lm.eq1}
are  tempered distributions.
Now the necessity. By  Corollary \ref{generalizedrieszomega1.cr} and Theorem \ref{maintheorem.iomega1}, $I-J_\Omega$
 is a continuous  linear operator from ${\mathcal S}$ to ${\mathcal S}'$ that satisfies that
$(I-J_\Omega)f=0$   for all $f\in {\mathcal S}_\infty$.
This implies that the inverse Fourier transform of the tempered distribution $(I-J_\Omega)^* g$
is supported on the origin for any Schwartz function $g$.
Hence there exist an integer $N$ and tempered distribution $H_{\bf i}, |{\bf i}|\le N$,
such that
${\mathcal F}^{-1}((I-J_\Omega)^* g)=\sum_{|{\bf i}|\le N} \langle g, H_{\bf i}\rangle \delta^{({\bf i})}/{{\bf i}!}$,
where the tempered distributions $\delta^{({\bf i})}, {\bf i}\in \ZZ_+^d$, are
defined by $\langle \delta^{({\bf i})}, f\rangle=\partial^{\bf i} f({\bf 0})$ \cite[Theorem 2.3.4]{hormanderbook}.
Then
$\langle (I-J_\Omega)f,  g\rangle=
\langle \hat f, {\mathcal F}^{-1}(I-J_\Omega)^*g\rangle= \sum_{|{\bf i}|\le N} \langle H_{\bf i}, g\rangle {\partial^{\bf i} \hat f({\bf 0})}/{{\bf i}!}
$ for all Schwartz functions $f$ and $g$, and hence \eqref{homogeneous1.lm.eq1} is established.
\end{proof}

\begin{Lm}\label{homogeneous2.lm}
Let  $\gamma$ be a positive number with $\gamma-d\not\in  \ZZ_+$,
 and consider the continuous linear operator $K$ from ${\mathcal S}$ to ${\mathcal S}'$:
 \begin{equation}\label{homogeneous2.lm.eq-1}
 Kf=\sum_{|{\bf i}|\le N} \frac{\partial^{\bf i} \hat f({\bf 0})}{{\bf i}!} H_{\bf i}, \quad f\in {\mathcal S}\end{equation}
  where $N\in {\mathbb Z}_+$ and   $H_{\bf i}, |{\bf i}|\le N$, are tempered distributions,
Then the following statements hold.
\begin{itemize}
\item[{(i)}] The  equation
\begin{equation} \label{homogeneous2.lm.eq0} K (\delta_t f)= t^{-\gamma} \delta_t (Kf)\end{equation}
 holds  for  any $f\in {\mathcal S}$ and $t>0$
 if and only if for every ${\bf i}\in \ZZ_+^d$ with $|{\bf i}|\le N$,
$H_{\bf i}$ is homogeneous of degree $\gamma-d-|{\bf i}|$.

\item[{(ii)}]
The linear operator  $K$  is translation-invariant if and only if
there exists a polynomial $P$ of degree at most $N$ such that
$H_{\bf i}=(-i\partial)^{\bf i} P$ for all ${\bf i}\in \ZZ_+^d$ with $|{\bf i}|\le N$.


\item[{(iii)}]  The linear operator $K$
  is translation-invariant and satisfies \eqref{homogeneous2.lm.eq0}
if and only if $H_{\bf i}=0$ for all
  ${\bf i}\in \ZZ_+^d$ with $|{\bf i}|\le N$.
\end{itemize}
\end{Lm}

\begin{proof}
{\em   (i)}\quad
The sufficiency follows from the homogeneous assumption on $H_{\bf i}, |{\bf i}|\le N$, and
 the observation that
\begin{equation}\label{homogeneous2.lm.pf.eq1}
\partial^{\bf i} \widehat {\delta_t f}({\bf 0})=t^{-d-|{\bf i}|} \partial^{\bf i} \hat f({\bf 0})\quad
{\rm for\ all}\  f\in {\mathcal S} \ {\rm and} \ {\bf i}\in \ZZ_+^d.\end{equation}
%
Now the necessity.
Let $\phi$ be a  $C^\infty$ function such that
   $\phi(\xi)=1$ for all $\xi\in B({\bf 0}, 1)$ and
   $\phi(\xi)=0$ for all $\xi\not\in B({\bf 0}, 2)$, where $B({\bf x}, r)$ is the  ball  with center ${\bf x}\in \RR^d$  and radius $r>0$.
   Define  $\psi_{\bf i}\in {\mathcal S}, {\bf i}\in \ZZ_+^d,$  with the help of the Fourier transform by
 \begin{equation}\label{homogeneous2.lm.pf.eq7}
 \widehat {\psi_{\bf i}}(\xi)=\frac{\xi^{\bf i}}{{\bf i} !} \phi(\xi). \end{equation}
One may  verify that
  \begin{equation}\label{homogeneous2.lm.pf.eq8}
  \partial^{{\bf i}'} \widehat{\psi_{\bf i}}({\bf 0})=\left\{\begin{array}{ll} 1 & {\rm if} \ {\bf i}'={\bf i},\\
  0 & {\rm if} \  {\bf i}'\ne {\bf i}.\end{array}\right.
  \end{equation}
  For any ${\bf i}\in \ZZ_+^d$ with $|{\bf i}|\le N$,
   the homogeneous property of the tempered distribution $H_{\bf i}$ follows by replacing $f$ in \eqref{homogeneous2.lm.eq0} by $\psi_{\bf i}$ and using \eqref{homogeneous2.lm.pf.eq8}.


 {\em (ii)}\quad ($\Longleftarrow$) Given  $f\in {\mathcal S}$ and ${\bf x}_0\in \RR^d$,
\begin{eqnarray} \label{homogeneous2.lm.pf.eq10}
 K(\tau_{{\bf x}_0}f) ({\bf x})  & = & \sum_{|{\bf i}|\le N}
 \sum_{{\bf j}+{\bf k}= {\bf i}}
 \frac{ (-i{\bf x}_0)^{{\bf k}}}{{\bf k}!} \frac{\partial^{\bf j} \hat f({\bf 0})}{{\bf j}!}
 (-i\partial)^{\bf i}P({\bf x})\nonumber\\
 &= & \sum_{|{\bf j}|\le N}
 \frac{(-i)^{\bf j}\partial^{\bf j} \hat f({\bf 0})}{{\bf j}!}
 \Big(\sum_{|{\bf k}|\le N-|{\bf j}|}\frac{\partial^{{\bf j}+{\bf k}} P({\bf x})}{{\bf k}!} (-{\bf x}_0)^{\bf k}\Big)
 \nonumber\\
 & = &  \sum_{|{\bf j}|\le N}
 \frac{(-i)^{\bf j}\partial^{\bf j} \hat f({\bf 0})}{{\bf j}!}\partial^{{\bf j}} P({\bf x}-{\bf x}_0)=Kf({\bf x}-{\bf x}_0),
  \end{eqnarray}
where the first equality follows from
 \begin{equation} \label{homogeneous2.lm.pf.eq11}
 \partial^{\bf i} \widehat {\tau_{{\bf x}_0}f}({\bf 0})=\sum_{{\bf j}\le {\bf i}}
 \binom{{\bf i}}{{\bf j}} (-i{\bf x}_0)^{{\bf i}-{\bf j}} \partial^{\bf j} \hat f({\bf 0}),
 \end{equation}
 and the third equality is deducted from the Taylor expression of the polynomial $\partial^{\bf j} P$
 of degree at most $N-|{\bf j}|$.

 ($\Longrightarrow$)
 By \eqref{homogeneous2.lm.pf.eq11} and the translation-invariance of the linear operator $K$,
  \begin{equation} \label{homogeneous2.lm.pf.eq12}
  \sum_{|{\bf i}|\le N}
 \sum_{{\bf j}+{\bf k}= {\bf i}}
 \frac{ (-i{\bf x}_0)^{{\bf k}}}{{\bf k}!} \frac{\partial^{\bf j} \hat f({\bf 0})}{{\bf j}!}
 H_{\bf i}= \sum_{|{\bf i}|\le N}
 \frac{\partial^{\bf j} \hat f({\bf 0})}{{\bf j}!} \tau_{ {\bf x}_0}H_{\bf j}
  \end{equation}
holds for any Schwartz function $f$ and ${\bf x}_0\in \RR^d$.
  Replacing $f$ in the above equation by the function $\psi_{\bf 0}$ in \eqref{homogeneous2.lm.pf.eq7}
  and then using \eqref{homogeneous2.lm.pf.eq8}, we get
  \begin{equation} \label{homogeneous2.lm.pf.eq13}
 \tau_{{\bf x}_0} H_{\bf 0} =
  \sum_{|{\bf i}|\le N}
 \frac{ (-i{\bf x}_0)^{{\bf i}}}{{\bf i}!}
 H_{\bf i}.
  \end{equation}
This implies that
 $ 
\langle H_{{\bf 0}}, g(\cdot+{\bf x}_0)\rangle=
\sum_{ |{\bf i}|\le N} \frac{ (-i{\bf x}_0)^{{\bf i}}}{{\bf i}!}
\langle H_{{\bf i}}, g\rangle
$ 
for any Schwartz function $g$. By taking partial derivatives $\partial^{\bf k}, |{\bf k}|=N+1$,  with respect to ${\bf x}_0$
of both sides of the above equation,
using the fact that $\partial^{{\bf k}} {\bf x}^{{\bf i}}=0$ for all ${\bf k}\in \ZZ_+$ with $|{\bf k}|=N+1$,
and then letting ${\bf x}_0={\bf 0}$, we obtain that
$\langle H_{{\bf 0}}, \partial^{\bf k} g\rangle =0
$ holds
for any $g\in {\mathcal S}$ and ${\bf k}\in \ZZ_+$ with $|{\bf k}|=N+1$.
Hence $H_{\bf 0}=P$ for some polynomial $P$ of degree at most $N$.
The desired conclusion about $H_{\bf i}, |{\bf i}|\le N$,
then follows from \eqref{homogeneous2.lm.pf.eq13} and
$\tau_{{\bf x}_0}H_{\bf 0}({\bf x})=\sum_{|{\bf i}|\le N} \frac{ (-{\bf x}_0)^{\bf i}}{{\bf i}!} \partial^{\bf i} P({\bf x})
$ by the  Taylor expansion of the polynomial $P$.

%

{\em (iii)} \quad Clearly if $H_{\bf i}=0$ for all $|{\bf i}|\le N$, then  $Kf=0$
for all $f\in {\mathcal S}$ and hence $K$ is  translation-invariant and satisfies \eqref{homogeneous2.lm.eq0}.
Conversely, if $K$ is translation-invariant and  satisfies \eqref{homogeneous2.lm.eq0}, it follow from
the conclusions (i) and (ii) that for every ${\bf i}\in \ZZ_+^d$ with $|{\bf i}|\le N$,
$H_{\bf i}$ is homogeneous of degree $\gamma-d-|{\bf i}|\not\in \ZZ$ and also a polynomial of degree at most $N-|{\bf i}|$.
Then $H_{\bf i}=0$ for all $ {\bf i}\in \ZZ_+^d$ with $|{\bf i}|\le N$ because
 the homogeneous degree of any  nonzero polynomial is a nonnegative integer if it is
 homogeneous.
\end{proof}

 We now have all of ingredients to prove Theorem \ref{iomega2.tm}.

\begin{proof}[Proof of Theorem \ref{iomega2.tm}] The sufficiency follows from  Corollary \ref{generalizedrieszomega1.cr} and Theorem \ref{maintheorem.iomega1}. Now the necessity.
By Lemma \ref{homogeneous1.lm}, there exist
an integer $N$ and tempered distributions $H_{\bf i}, |{\bf i}|\le N$,
such that   \eqref{homogeneous1.lm.eq1} holds.
Define
 $Kf=\sum_{|{\bf i}|\le N} \frac{\partial^{\bf i} \hat f({\bf 0})}{{\bf i}!} H_{\bf i}$ for any $f\in {\mathcal S}$. Then
 $Kf $ is a  continuous linear  operator from ${\mathcal S}$ to ${\mathcal S}'$ and
 \begin{equation}\label{iomega2.tm.pf.eq1}
 If=J_\Omega f+Kf,  \quad \ f\in {\mathcal S}.\end{equation}
  Moreover the linear operator $K$ satisfies \eqref{homogeneous2.lm.eq0} and  is   translation-invariant  by  \eqref{iomega2.tm.pf.eq1},
  Theorem \ref{maintheorem.iomega1}  and the assumption on $I$.
  Then $Kf=0$ for all $f\in {\mathcal S}$ by Lemma \ref{homogeneous2.lm}. This together with
  \eqref{iomega2.tm.pf.eq1} proves the desired conclusion that $I=J_\Omega$.
  %
%
%
%
%
%
\end{proof}

\subsection{Translation-invariant extensions of the linear operator $i_\Omega$ with additional localization in  the Fourier domain}
 Given a  nonzero homogeneous function $\Omega\in C^\infty(\RR^d\backslash \{\bf 0\})$ of degree
 $-\gamma$, we recall from \eqref{fractionalderivative.neweq2} and
 Theorem \ref{maintheorem.iomega1} that $J_\Omega$ is translation-invariant and
 the Fourier transform of $J_\Omega f$ belongs to $K_1$ when $\gamma\in (0, d)$,
 where
 \begin{equation}
K_1=\Big\{h:\ \int_{\RR^d} |h(\xi)| (1+\|\xi\|)^{-N} d\xi<\infty \ \ {\rm for \ some} \ N\ge 1\Big\}.\end{equation}
In fact, the generalized Riesz potential $J_\Omega$ is the {\bf only} extension
  of  the linear operator $i_\Omega$ on ${\mathcal S}_\infty$ to the whole space ${\mathcal S}$ with the above two properties.

  \begin{Tm}\label{iomega4.tm}
  Let $\gamma>0$ with $\gamma-d\not\in \ZZ_+$,  $\Omega\in C^\infty(\RR^d\backslash \{\bf 0\})$ be
 a nonzero homogeneous function of degree $-\gamma$,
 and the continuous linear operator $I$ from ${\mathcal S}$ to ${\mathcal S}'$
  be an extension of the  linear operator $i_\Omega$ on  ${\mathcal S}_\infty$
   such that the Fourier transform
  of $If$ belongs to $K_1$ for all $f\in {\mathcal S}$.
  Then  $I$ is translation-invariant if and only if  $I=J_\Omega$ and $\gamma\in (0, d)$.
  \end{Tm}

\begin{proof} The sufficiency  follows from \eqref{fractionalderivative.neweq2} and
 Theorem \ref{maintheorem.iomega1}. Now we prove the necessity.
 By the assumption on the linear operator $I$, applying an argument  similar to the proof of
 Lemma \ref{homogeneous1.lm}, we can find  a family of  functions $g_{\bf i}\in K_1, |{\bf i}|\le N$,
  such that
  \begin{eqnarray}\label{iomega4.tm.pf.eq0}
\widehat{If}(\xi) & =  &  \Big(\hat f(\xi)-
 \sum_{|{\bf i}|\le \gamma-d} \frac{\partial^{{\bf i}} \hat f({\bf 0})}{ {\bf i}!} \xi^{\bf i}\Big)
\Omega(\xi)  +\sum_{|{\bf i}|\le N} \frac{\partial^{{\bf i}} \hat f({\bf 0})}{ {\bf i}!} g_{\bf i}(\xi)
\end{eqnarray}
for any Schwartz function $f$.
This together with \eqref{homogeneous2.lm.pf.eq11} and the translation-invariance of the linear operator $I$
implies that
\begin{eqnarray*}\label{iomega4.tm.pf.eq1}
 & & -
\sum_{ |{\bf i}|\le \gamma-d} \sum_{{\bf j}+{\bf k}={\bf i}}
\frac{\partial^{\bf j}\hat f({\bf 0})}{{\bf k}! {\bf j}!} (-i{\bf x}_0)^{{\bf k}} \xi^{\bf i}\Omega(\xi)  +
\sum_{ |{\bf i}|\le  N} \sum_{{\bf j}+{\bf k}= {\bf i}}
\frac{\partial^{\bf j}\hat f({\bf 0})}{{\bf k}! {\bf j}!} (-i{\bf x}_0)^{{\bf k}} g_{\bf i}(\xi)\nonumber\\
& = & \
e^{i{\bf x}_0\xi}
\Big(-\sum_{|{\bf i}|\le \gamma-d} \frac{\partial^{\bf i}\hat f({\bf 0})}{{\bf i}!} \xi^{\bf i}\Omega(\xi)
 +\sum_{|{\bf i}|\le  N}
\frac{\partial^{\bf i}\hat f({\bf 0})}{{\bf i}!}  g_{\bf i}(\xi)\Big).
\end{eqnarray*}
As ${\bf x}_0\in \RR^d$ in \eqref{iomega4.tm.pf.eq1} is chosen arbitrarily, we conclude that
\begin{equation*}\label{iomega4.tm.pf.eq2}
-\sum_{|{\bf i}|\le \gamma-d} \frac{\partial^{\bf i}\hat f({\bf 0})}{{\bf i}!} \xi^{\bf i}\Omega(\xi)
 +\sum_{|{\bf i}|\le  N}
\frac{\partial^{\bf i}\hat f({\bf 0})}{{\bf i}!}  g_{\bf i}(\xi)=0\quad {\rm for \ all} \ f\in {\mathcal S}.\end{equation*}
Substituting the above equation into \eqref{iomega4.tm.pf.eq0}, we then obtain
$\widehat{If}(\xi)= \hat f(\xi) \Omega(\xi)
$
for all $f \in {\mathcal S}$. This, together with
 the observation that
 $\hat f\Omega\in K_1$ for all $f\in {\mathcal S}$ if and only if $\gamma<d$, leads to the desired conclusion that
 $I=J_\Omega$ and $\gamma\in (0, d)$.
\end{proof}

\subsection{Non-integrability  in the spatial domain}
Let $\gamma>0$ with $\gamma-d\not\in \ZZ_+$ and
 $\Omega\in C^\infty(\RR^d\backslash \{\bf 0\})$ be a nonzero
 homogeneous function of degree $-\gamma$. For any Schwartz function $f$, there exists a positive constant $C$ by
  Theorem \ref{generalizedrieszomega1.tm} such that  $|J_\Omega f({\bf x})|\le C (1+\|{\bf x}\|)^{\gamma-d}$
  for all ${\bf x}\in \RR^d$. Hence  $J_\Omega f\in L^p, 1\le p\le \infty$, when $\gamma<d(1-1/p)$.
In this subsection, we show that the above $p$-integrability property for  the generalized Riesz potential $J_\Omega$
 is no longer true when $\gamma\ge d(1-1/p)$.

\begin{Tm}\label{time1.tm}
Let $1\le p\le \infty, 0<\gamma\in [d(1-1/p), \infty)\backslash \ZZ$ and
 $\Omega\in C^\infty (\RR^d\backslash \{\bf 0\})$ be a nonzero homogeneous function of degree $-\gamma$.
 Then there exists a Schwartz function $f$ such that
$J_\Omega f\not\in L^p$.
\end{Tm}

Letting $\Omega(\xi)=\|{\xi}\|^{-\gamma}$  in Theorem \ref{time1.tm} leads to the  conclusion mentioned in the abstract:

\begin{Cr}\label{nonintegrable.cr}
Let $1\le p\le \infty$ and $d(1-1/p)\le \gamma\not\in \ZZ_+$.
 Then
$I_\gamma f$ is {\bf not} $p$-integrable for some function $f\in {\mathcal S}$.
\end{Cr}


\begin{proof}[Proof of Theorem \ref{time1.tm}]
Let the Schwartz functions $\phi$  and $\psi_{\bf i}, {\bf i}\in \ZZ_+^d$, be as in the proof of Lemma \ref{homogeneous2.lm}.
We examine three cases to prove the theorem.

{\em Case I:\quad $d(1-1/p)\le\gamma< \min(d, d(1-1/p)+1)$}.\quad
 In this case, $1\le p<\infty$ and
\begin{equation}\label{time1.tm.pf.eq1}
J_\Omega \psi_{\bf 0}({\bf x})=\int_{\RR^d} K({\bf x}-{\bf y}) \psi_{\bf 0}({\bf y}) d{\bf y},
\end{equation}
 by \eqref{fractionalderivative.neweq2}, where $K$ is the inverse Fourier transform of $\Omega$.
By  \cite[Theorems  7.1.16 and 7.1.18]{hormanderbook},
 $K\in C^\infty(\RR^d\backslash\{0\})$
is a homogeneous function
of order $\gamma-d\in (-d,0)$, which  implies that
\begin{equation}\label{time1.tm.pf.eq2}
|\partial^{\bf i} K({\bf x})|\le C \|{\bf x}\|^{\gamma-d-|{\bf i}|}  \quad {\rm for \ all}\ {\bf i}\in \ZZ_+^d \ {\rm with} \ |{\bf i}|\le 1.
\end{equation}
Using \eqref{time1.tm.pf.eq1} and \eqref{time1.tm.pf.eq2}, and
noting that $\psi_{\bf 0}\in {\mathcal S}$ satisfies
$\int_{\RR^d} \psi_{\bf 0}({\bf y})d{\bf y}=1$, we obtain
that for all ${\bf x}\in \RR^d$ with $\|{\bf x}\|\ge 1$,
\begin{eqnarray}\label{time1.tm.pf.eq3}
   |J_\Omega \psi_{\bf 0} ({\bf x})-K({\bf x})|
\!\! & \le &\!\!
\int_{\|{\bf y}\|\le \|{\bf x}\|/2} |K({\bf x}-{\bf y})-K({\bf x})| |\psi_{\bf 0}({\bf y})| d{\bf y} \nonumber\\
 & &\!\! +
\Big(\int_{\|{\bf x}\|/2\le  \|{\bf y}\|\le 2 \|{\bf x}\|}+\int_{2\|{\bf x}\|\le  \|{\bf y}\|}\Big) |K({\bf x}-{\bf y})||\psi_{\bf 0}({\bf y})| d{\bf y}
 \nonumber\\
 & &\!\! + |K({\bf x})| \int_{\|{\bf y}\|\ge  \|{\bf x}\|/2} |\psi_{\bf 0}({\bf y})| d{\bf y}
\nonumber\\ 
\!\! & \le &\!\! C (1+\|{\bf x}\|)^{\gamma-d-1}.
\end{eqnarray}
We notice that   $\int_{\|{\bf x}\|\ge 1}  (1+\|{\bf x}\|)^{(\gamma-d-1)p}d {\bf x}<\infty$
 and
$\int_{\|{\bf x}\|\ge 1} |K({\bf x})|^p d{\bf x}=\infty$ because
 $K$ is a nonzero homogenous function
of degree $\gamma-d$ and  $d-p<(d-\gamma)p\le d$.
The above two observations together with the estimate in \eqref{time1.tm.pf.eq3}
prove that $J_\Omega \psi_{\bf 0} \not\in L^p$, the desired conclusion  with $f=\psi_{\bf 0}$.

{\em Case II: $d<\gamma<d(1-1/p)+1$}.\quad   In this case, $d<p\le \infty$ and
\begin{eqnarray}\label{time1.tm.pf.eq4}
J_\Omega \psi_{\bf 0}({\bf x}) & = &\frac{1}{d-\gamma}
 \sum_{|{\bf j}|=1} J_{\Omega_{\bf j}}(\varphi_{\bf j}) ({\bf x})
+ \frac{1}{d-\gamma}\sum_{|{\bf i}|=1} (-{\bf x})^{\bf i} J_{\Omega_{\bf i}}\psi_{\bf 0} ({\bf x})
\end{eqnarray}
by taking $k_0=1$ in \eqref{generalizedrieszomega1.tm.pf.eq1}, where $\Omega_{\bf i}(\xi)=(i\xi)^{\bf i} \Omega(\xi)$ and
$\varphi_{\bf i}({\bf x})={\bf x}^{\bf i} \psi_{\bf 0}({\bf x})$. Let $K_{\bf i}$ be the inverse
 Fourier transform
of the function $\Omega_{\bf i}, |{\bf i}|=1$. Noticing that $\Omega_{\bf i}$ is homogeneous of degree $-\gamma+1$ and that $\int_{\RR^d} \varphi_{\bf i}({\bf x}) d{\bf x}=0$, we then
apply similar argument to the one used in establishing
\eqref{time1.tm.pf.eq3}  and obtain
\begin{eqnarray*}
|J_{\Omega_{\bf i}}(\varphi_{\bf i}) ({\bf x})|+
|J_{\Omega_{\bf i}}\psi_{\bf 0} ({\bf x})-K_{\bf i}({\bf x})|
\le C \|{\bf x}\|^{\gamma-d-2}\quad  {\rm if } \ \|{\bf x}\|\ge 1.
\end{eqnarray*}
 Hence
\begin{equation}
\label{time1.tm.pf.eq3b}
\int_{\|{\bf x}\|\ge 1} \big|J_\Omega \psi_{\bf 0}({\bf x})-\frac{1}{d-\gamma}\sum_{|{\bf i}|=1} (-{\bf x})^{\bf i} K_{\bf i}({\bf x})\big|^p d{\bf x}\le
C \int_{\|{\bf x}\|\ge 1} \|{\bf x}\|^{(\gamma-d-1)p} d{\bf x}<\infty
\end{equation}
if $d<p<\infty$ and
\begin{equation}\label{time1.tm.pf.eq4}
\sup_{\|{\bf x}\|\ge 1} \big|J_\Omega \psi_{\bf 0}({\bf x})-\frac{1}{d-\gamma}\sum_{|{\bf i}|=1} (-i{\bf x})^{\bf i} K_{\bf i}({\bf x})\big|\le
C \sup_{\|{\bf x}\|\ge 1} \|{\bf x}\|^{\gamma-d-1} <\infty
\end{equation}
if $p=\infty$.
Set $K({\bf x}):= \sum_{|{\bf i}|=1} (-{\bf x})^{\bf i} K_{\bf i}({\bf x})$.
Then  $K$ is
homogeneous of degree $\gamma-d$ by the assumption on $\Omega$, and  is not identically zero because
\begin{eqnarray*}
\langle K, g\rangle & = &  \int_{\RR^d}
\Omega(\xi) \Big(\sum_{|{\bf i}|=1} \xi^{\bf i} \partial^{\bf i} \hat g(\xi)\Big) d\xi
 =  -\int_{\RR^d} \Big(\sum_{|{\bf i}|=1}\partial^{\bf i}( \xi^{\bf i} \Omega(\xi))\Big) \hat g(\xi) d\xi \\ 
& = & \int_{S^{d-1}} \int_0^\infty \big(d\Omega(r\xi')+ r \frac{d}{dr} \Omega(r\xi')\big) \hat g(r\xi') r^{d-1} dr d\sigma(\xi')\\
&
 =&  (d-\gamma) \int_{\RR^d} \Omega(\xi) \hat g (\xi) d\xi
\not\equiv 0
\end{eqnarray*} where $g\in {\mathcal S}_\infty$.
Thus
$\int_{\|{\bf x}\|\ge 1} |K({\bf x})|^p d{\bf x}=+\infty$ when $d<p<\infty$, and
$K({\bf x})$ is unbounded on $\RR^d\backslash B({\bf 0},1)$ when $p=\infty$.
This together with \eqref{time1.tm.pf.eq3b} and  \eqref{time1.tm.pf.eq4}
proves that   $J_\Omega \psi_{\bf 0}\not\in L^p$ and hence the desired conclusion with $f=\psi_{\bf 0}$.

{\em Case III: $\gamma\ge d(1-1/p)+1$.}
Let $k_0$ be  the integer such that $d(1-1/p)\le \gamma-k_0<d(1-1/p)+1$, and  set
$\Omega_{\bf j}(\xi)=(i\xi)^{\bf j} \Omega(\xi), |{\bf j}|=k_0$.
Noting that
$J_\Omega \psi_{\bf j}({\bf x})=
 J_{\Omega_{\bf j}} \psi_{\bf 0} ({\bf x})/{{\bf j}!}$
 and $\Omega_{\bf j}$ is  homogeneous of degree $-\gamma+k_0$,
 we have obtained from the conclusions in the first two cases that $J_\Omega \psi_{\bf j} \not\in L^p$.
 Hence the desired conclusion follows by letting $f=\psi_{\bf j}$ with $|{\bf j}|=k_0$. \end{proof}

\subsection{Non-integrability  in the Fourier domain}
If $\gamma<d$, it follows from   \eqref{fractionalderivative.neweq2}
 that for  Schwartz functions $f$ and $g$,
 $\langle J_\Omega f, g\rangle$ can be expressed as
 a weighted integral of  $\hat g$:
\begin{equation}\label{frequency.eq1}
\langle J_\Omega f, g\rangle=\int_{\RR^d} h(\xi) \hat g(\xi) d\xi,
\end{equation}
where $h(\xi)=(2\pi)^{-d}\Omega(-\xi) \hat f(-\xi)\in K_1$.
In this subsection, we show that the above reformulation \eqref{frequency.eq1}
to  define  $\langle J_\Omega f, g\rangle$ via  a  weighted integral of $\hat g$  {\bf cannot} be extended to  $\gamma>d$.

\begin{Tm}\label{frequency.tm1}
Let $\gamma\in (d, \infty)\backslash \ZZ$,
 $\Omega\in C^\infty(\RR^d\backslash \{\bf 0\})$ be a nonzero homogeneous function of degree $-\gamma$,
 and let $J_\Omega$ be defined by \eqref{fractionalderivative.def}.
 Then there exists a  Schwartz function $f$ such that the Fourier transform of $J_\Omega f$ does not belong to $K_1$.
\end{Tm}

\begin{proof}
Let  $\phi$ and $\psi_{\bf 0}$ be the Schwartz functions in the proof of Lemma \ref{homogeneous2.lm}, and
let $g\in {\mathcal S}_\infty$ be so chosen that
its  Fourier transform
 $\hat g$ is supported in $B({\bf 0}, 1)$
   and satisfies $\int_{\RR^d} \Omega(\xi){ \hat g(-\xi)}d\xi=1$. Now we prove that
    $\widehat{J_\Omega \psi_{\bf 0}}\not\in K_1$.
 Suppose on the contrary that $\widehat{J_\Omega \psi_{\bf 0}}\in K_1$.
 Then
   \begin{eqnarray}\label{frequency.tm1.pf.eq1}   \langle J_\Omega \psi_{\bf 0}, n^{-d} g(\cdot/n)\rangle
 & = &\frac{ (2\pi)^{-d}\Gamma(d-\gamma)} {\Gamma(d+k_0-\gamma)}
 \int_{S^{d-1}}\int_\epsilon^\infty
 r^{k_0+d-\gamma-1} \Omega(\xi')\nonumber\\
 & &  \Big(-\frac{d}{dr}\Big)^{k_0}\Big(\widehat \psi_{\bf 0}(r\xi') \hat g(-rn\xi')\Big) dr d\sigma(\xi') \nonumber\\
   & = & (2\pi)^{-d} \int_{\RR^d} 
   {\hat g(-n\xi)} \Omega(\xi)d\xi
  =  (2\pi)^{-d} n^{\gamma-d}\int_{\RR^d} \Omega(\xi) {\hat g(-\xi)} d\xi \nonumber\\
 &\to&  +\infty \quad {\rm as} \ n\to \infty
\end{eqnarray}
by \eqref{fractionalderivative.def} and \eqref{extension.eq}.
On the other hand,
  \begin{eqnarray}\label{frequency.tm1.pf.eq2}
& &   |\langle J_\Omega \psi_{\bf 0}, n^{-d} g(\cdot/n)\rangle|
    = (2\pi)^{-d} \Big|\int_{\RR^d} \widehat{J_\Omega \psi_{\bf 0}}(\xi) \hat g(-n\xi) d\xi\Big|\nonumber\\
&   \le  &  (2\pi)^{-d} \|\hat g\|_\infty \int_{|\xi|\le 1/n} |\widehat{J_\Omega \psi_0}(\xi)| d\xi
    \to  0 \ {\rm as} \ n\to \infty,
  \end{eqnarray}
  where we have used the hypothesis that $\widehat{J_\Omega \psi_{\bf 0}}\in K_1$ to obtain the limit.
  The limits in \eqref{frequency.tm1.pf.eq1} and \eqref{frequency.tm1.pf.eq2} contradict each other, and hence
  the Fourier transform $J_\Omega \psi_{\bf 0}$ does not belong to $K_1$.
\end{proof}

  \subsection{Proof of Theorem \ref{generalizedriesz.tm}}
  Observe that $J_\Omega=I_\gamma$ when $\Omega(\xi)=\|\xi\|^{-\gamma}$ and $\gamma>0$, and that
\begin{equation}
\label{generalized.tm.pf.eq1} J_\Omega=(-\triangle)^{-\gamma/2}\quad {\rm if}\ \Omega(\xi)=\|\xi\|^{-\gamma}\quad {\rm  and} \ \gamma<0.
\end{equation} Then
 the necessity holds by Theorem \ref{iomega2.tm}, while the sufficiency follows from Corollary \ref{generalizedrieszomega1.cr}, Theorem \ref{maintheorem.iomega1}, and Corollary \ref{composition.cr}.

\section{Integrable  Riesz Potentials}\label{irp.section}


In Section \ref{grp.section}, we have shown that the various attempts for defining a proper (integrable) Riesz potential that is translation-invariant
are doomed to failure for $\gamma>d$. We now proceed by providing a fix which is possible if we drop the translation-invariance requirement.

Let $1\le p\le \infty, \gamma\in \RR$, and  $\Omega\in C^\infty(\RR^d\backslash \{\bf 0\})$
be a homogeneous function
 of degree $-\gamma$. We define the linear operator  $U_{\Omega,p} $ from ${\mathcal S}$ to ${\mathcal S}'$
  with the help of the Fourier transform by
 \begin{equation}\label{newfractionalderivative.def}
{\mathcal F}({U_{\Omega,p}f})(\xi)= \Big(\hat f (\xi)-\sum_{|{\bf i}|\le \gamma-d(1-1/p)} \frac{\partial^{\bf i}\hat f({\bf 0})}{ {\bf i}!} \xi^{\bf i}\Big) \Omega(\xi), \quad f\in {\mathcal S}.\end{equation}
We call the linear operator $U_{\Omega, p}$ a {\em
$p$-integrable Riesz potential associated with  the homogenous function $\Omega$}, or {\em
integrable Riesz potential} for brevity, as
\begin{equation} \label{fractionalderivativeomegap.def}U_{\Omega, p}=I_{\gamma, p} \quad {\rm if } \quad \Omega(\xi)=\|\xi\|^{-\gamma}.\end{equation}

   Define
 \begin{equation}\label{fractionalderivative.tm1.pf.eq3}  
  U_{\Omega, p}^*f({\bf x})   = (2\pi)^{-d}
\int_{\RR^d}\Big(e^{i\langle {\bf x}, \xi\rangle}-\sum_{|{\bf i}|\le \gamma-d+d/p}
  \frac{(i{\bf x})^{\bf i}\xi^{\bf i}}{ {\bf i}!} \Big)
{\Omega(-\xi)}\hat f(\xi) d\xi, \quad    f\in {\mathcal S}.
 \end{equation}
Then  $U_{\Omega, p}^*$ is the adjoint operator of the integrable Riesz potenrial $U_{\Omega, p}$:
\begin{equation}\label{fractionalderivative.tm1.pf.eq4}
\langle U_{\Omega, p}f, g\rangle=\langle f, U_{\Omega, p}^*g\rangle \quad {\rm for\ all} \ f, g\in {\mathcal S}.
\end{equation}

If $\gamma$ satisfies $0<\gamma< d(1-1/p)$, then
\begin{equation} U_{\Omega, p}f= J_\Omega f \quad  {\rm for \ all} \  f\in {\mathcal S}.\end{equation}
Hence in this case, it follows from  Theorem \ref{maintheorem.iomega1} that
  $U_{\Omega, p}$ is dilation-invariant and translation-invariant,  and a continuous
 extension of the linear operator $i_\Omega$  on the closed subspace ${\mathcal S}_\infty$ to the whole space ${\mathcal S}$.
 Moreover
  $U_{\Omega, p} f\in L^p$  and  ${\mathcal F}( {U_{\Omega, p}f})\in L^q, 1\le q\le p/(p-1)$, for any Schwartz function $f$
  by  Theorem \ref{generalizedrieszomega1.tm} and the following estimate:
 $$
   |{\mathcal F}({U_{\Omega, p}f})(\xi)|\le C \|\xi\|^{-\gamma} (1+\|\xi\|)^{\gamma-d-1}\ {\rm for \ all} \ \xi\in \RR^d.$$
So from now on, we implicitly assume that $\gamma\ge d(1-1/p)$,
 except when mentioned otherwise.

In the sequel,
we investigate with the properties of the $p$-integrable Riesz potential $U_{\Omega, p}$
associated with  a homogenous function $\Omega$, such as
 dilation-invariance and translation-variance (Theorem \ref{fractionalderivative.tm1}),
$L^{p/(p-1)}$-integrability in the Fourier domain (Corollary \ref{fractionalderivative.cr1}),  $L^{p}$-integrability in  the spatial domain (Theorem \ref{iomegap.tm1} and Corollary \ref{iomegap.cor1}), composition and  left-inverse property (Theorem \ref{compositionp.tm} and Corollary \ref{leftinversefractionalderivative.cr}),
the uniqueness of dilation-invariant extension of the linear operator $i_\Omega$ from the closed subspace ${\mathcal S}_\infty$ to the whole space ${\mathcal S}$ with additional integrability in the spatial domain and in the Fourier domain (Theorems \ref{time2.tm} and \ref{iomega5.tm}).
The above properties of the  $p$-integrable Riesz potential
associated with a homogenous function will be used  to prove Theorem \ref{integrablefractionalderivative.tm} in the last subsection.

\subsection{Dilation-invariance, translation-variance and integrability in the Fourier domain}

\begin{Tm}\label{fractionalderivative.tm1} Let $1\le p\le \infty, \gamma\ge d(1-1/p)$, $k_1$ be the integral part of
 $\gamma-d(1-1/p)$, $\Omega\in C^\infty(\RR^d\backslash \{{\bf 0}\})$ be a  nonzero homogeneous function
of degree $-\gamma$, and  let $U_{\Omega, p}$ be defined as in \eqref{newfractionalderivative.def}.
Then the following statements hold.
\begin{itemize}
\item[{(i)}]  $U_{\Omega,p}$ is dilation-invariant.


\item[{(ii)}]  $U_{\Omega, p}$ is not translation-invariant.

\item[{(iii)}] If $\sup_{{\bf x}\in \RR^d} |f({\bf x})| (1+\|{\bf x}\|)^{k_1+d+1+\epsilon}<\infty$ for some $\epsilon>0$, then there exists a positive constant $C$  independent on $f$ such that
\begin{equation}\label{fractionalderivative.tm1.eq1}
|{\mathcal F}({U_{\Omega, p} f})(\xi)|\le C \Big(\sup_{{\bf z}\in \RR^d} |f({\bf z})| (1+\|{\bf z}\|)^{k_1+d+1+\epsilon} \Big) \|\xi\|^{k_1-\gamma+1} (1+\|\xi\|)^{-1}\end{equation}
for all $\xi\in \RR^d$.

\item[{(iv)}]  $U_{\Omega,p}$ is a continuous linear operator from ${\mathcal S}$ to ${\mathcal S}'$, and an  extension of the operator $i_\Omega$ on the subspace ${\mathcal S}_\infty$ to the whole space
    ${\mathcal S}$.

\end{itemize}
\end{Tm}

As a consequence of Theorem \ref{fractionalderivative.tm1}, we have the following result about the $L^{p/(p-1)}$-integrability of the Fourier transform
of $U_{\Omega, p} f$ for $f\in {\mathcal S}$.

\begin{Cr}\label{fractionalderivative.cr1}
Let $1\le p\le \infty$ and $\gamma\ge  d(1-1/p)$ satisfy either $p=1$ or  $\gamma-d(1-1/p)\not\in \ZZ_+$ and $1<p\le \infty$, $k_1$ be the integral part of
 $\gamma-d(1-1/p)$, $\Omega\in C^\infty(\RR^d\backslash \{{\bf 0}\})$ be a  homogeneous function
of degree $-\gamma$, and  let $U_{\Omega, p}$ be defined as in \eqref{newfractionalderivative.def}.
Then
 the Fourier transform of $U_{\Omega, p} f$  belongs to $L^{p/(p-1)}$ for any $f\in {\mathcal S}$.
    \end{Cr}

\begin{proof} [Proof of Theorem \ref{fractionalderivative.tm1}]
{\em (i)}\quad  Given any $t>0$ and $f\in {\mathcal S}$,
\begin{equation*}\label{fractionalderivative.tm1.pf.eq3}
{\mathcal F}(U_{\Omega, p} (\delta_t f))(\xi)  =  t^{-d}
\Big(\hat f\big(\frac{\xi}{t}\big)-\sum_{|{\bf i}|\le \gamma-d+d/p} \frac{\partial^{\bf i} \hat f({\bf 0})}{{\bf i}!} \big(\frac{\xi}{t}\big)^{\bf i}\Big)  \Omega(\xi) =t^{-d-\gamma}
  {\mathcal F}( {U_{\Omega, p} f})\big(\frac{\xi}{t}\big).
\end{equation*}
This proves the dilation-invariance of the linear operator  $U_{\Omega, p}$.


{\em (ii)}\quad Suppose, on the contrary, that $U_{\Omega,p}$ is translation-invariant. Then
\begin{equation} \label{fractionalderivative.tm1.pf.eq5} \Omega(\xi)
\sum_{|{\bf i}|\le \gamma-d+d/p}
\frac{\partial^{\bf i} \widehat {\tau_{{\bf x}_0}f}({\bf 0})}{{\bf i}!} \xi^{\bf i}=
 \Omega(\xi) e^{-i\langle {\bf x}_0, \xi\rangle} \sum_{|{\bf i}|\le \gamma-d+d/p}
\frac{\partial^{\bf i} \hat f({\bf 0})}{{\bf i}!} \xi^{\bf i},\quad \xi\in \RR^d
\end{equation}
for all ${\bf x}_0\in \RR^d$ and $f\in {\mathcal S}$.
Note that the left-hand side of  equation \eqref{fractionalderivative.tm1.pf.eq5} is a polynomial in ${\bf x}_0$ by
\eqref{homogeneous2.lm.pf.eq11} while its right hand side  is a trigonometric function of ${\bf x}_0$. Hence
both sides must be identically zero, which implies that
\begin{equation}\Omega(\xi)
\sum_{|{\bf i}|\le \gamma-d+d/p}
\frac{\partial^{\bf i} \hat f({\bf 0})}{{\bf i}!} \xi^{\bf i}=0, \quad \xi\in \RR^d
\end{equation}
for all $f\in {\mathcal S}$. Replacing $f$ in the above equation by the function $\psi_{\bf 0}$ in \eqref{homogeneous2.lm.pf.eq7} and using
\eqref{homogeneous2.lm.pf.eq8} and the assumption $\gamma\ge d(1-1/p)$ leads to a contradiction.

{\em (iii)}\quad
By the assumption on the homogeneous function $\Omega$,
 $|\Omega(\xi)|\le C \|\xi\|^{-\gamma}$. Then
for $\xi\in \RR^d$ with $\|\xi\|\ge 1$,
\begin{eqnarray*}
|{\mathcal F}({U_{\Omega,p} f})(\xi)| & \le &
C \Big(\|\hat f\|_\infty+\sum_{|{\bf i}|\le k_1}
\|\partial^{\bf i}\hat f\|_\infty \|\xi\|^{|\bf i|}\Big)
\|\xi\|^{-\gamma}\nonumber\\
&  \le &   C \Big(\sum_{|{\bf i}|\le k_1+1} \|\partial^{\bf i} \hat f\|_\infty\Big)
\|\xi\|^{k_1-\gamma}
\end{eqnarray*}
by \eqref{newfractionalderivative.def},
and for $\xi\in \RR^d$ with $\|\xi\|\le 1$,
\begin{eqnarray*}
|{\mathcal F}({U_{\Omega,p} f})(\xi)| & \le &
C \Big(\sum_{|{\bf i}|\le k_1+1} \|\partial^{\bf i} \hat f\|_\infty\Big)
\|\xi\|^{k_1-\gamma+1}
\end{eqnarray*}
by the Taylor's expansion to  the function $\hat f(\xi)$ at the origin.
Combining the above two estimates gives
 \begin{equation}\label{fractionalderivative.tm1.pf.eq1}
 |{\mathcal F}({U_{\Omega,p} f})(\xi)|\le C
 \Big(\sum_{|{\bf i}|\le k_1+1} \|\partial^{\bf i} \hat f\|_\infty\Big) \|\xi\|^{k_1-\gamma+1} (1+\|\xi\|)^{-1}, \quad \xi\in \RR^d.
 \end{equation}
 Note that
 \begin{equation}\label{fractionalderivative.tm1.pf.eq2}
 \|\partial^{\bf i} \hat f\|_\infty\le
  C   \int_{\RR^d} |f({\bf x})| |{\bf x}|^{|{\bf i}|} d{\bf x}\le C
\sup_{{\bf z}\in \RR^d} |f({\bf z})| (1+|{\bf z}|)^{k_1+d+1+\epsilon}
 \end{equation}
 for all ${\bf i}\in \ZZ_+^d$ with $|{\bf i}|\le k_1+1$.
 Then the desired estimate \eqref{fractionalderivative.tm1.eq1} follows from
 \eqref{fractionalderivative.tm1.pf.eq1} and \eqref{fractionalderivative.tm1.pf.eq2}.

{\em (iv)}\quad By \eqref{newfractionalderivative.def} and the first conclusion of this theorem,
the Fourier transform of $U_{\Omega, p} f$ is continuous on $\RR^d\backslash \{{\bf 0}\}$, and satisfies
$$\int_{\RR^d} |{\mathcal F}({U_{\Omega,p} f})(\xi)| (1+\|\xi\|)^{\gamma-k_1-d-1}
d\xi\le C  \sup_{{\bf z}\in \RR^d} |f({\bf x})| (1+\|{\bf x}\|)^{k_1+d+2}.$$
Hence $U_{\Omega, p}$ is a continuous linear operator from ${\mathcal S}$ to ${\mathcal S}'$.
  For any $f\in {\mathcal S}_\infty$, $\partial^{\bf i}\hat f({\bf 0})=0$ for all ${\bf i}\in \ZZ_+^d$. Then
${\mathcal F}({U_{\Omega,p} f})={\mathcal F}({i_\Omega f})$ for all $f\in {\mathcal S}_\infty$.
This shows that $U_{\Omega, p}, 1\le p\le \infty,$ is a
 continuous extension of the  linear operator $i_\Omega$
from the subspace ${\mathcal S}_\infty\subset {\mathcal S}$ to the whole space ${\mathcal S}$.
\end{proof}

\subsection{Composition  and left-inverse of the fractional Laplacian}
Direct calculation leads to
$$\sum_{|{\bf i}|\le \gamma-d(1-1/p)}\frac{\partial^{\bf i} (\xi^{\bf k}\hat f(\xi))|_{\xi={\bf 0}}}{{\bf i}!}\xi^{\bf i}
 =
\sum_{|{\bf j}|\le \gamma-|{\bf k}|-d(1-1/p)}
\frac{\partial^{\bf j} \hat f({\bf 0})}{{\bf j}!}\xi^{{\bf j}+{\bf k}}, \quad {\bf k}\in \ZZ_+^d$$
for any $\gamma\in \RR, 1\le p\le \infty$ and $f\in {\mathcal S}$.
This together with \eqref{newfractionalderivative.def} implies that
\begin{equation}
U_{\Omega, p} (\partial^{\bf k} f)= U_{\Omega_{\bf k}, p} f, \quad {\rm for \ all} \ f\in {\mathcal S}\ {\rm and} \ {\bf k}\in \ZZ_+^d,
\end{equation}
 where $\Omega_{\bf k}(\xi)= (i\xi)^{\bf k} \Omega(\xi)$  for ${\bf k}\in \ZZ_+^d$. In general, we have the following result about composition of integrable Riesz potentials.

\begin{Tm}\label{compositionp.tm}
Let $1\le p\le \infty$, real numbers $\gamma_1, \gamma_2$  satisfy $\gamma_1\ge d(1-1/p)$
and $-\gamma_2$ is larger than the integral part of $\gamma_1-d(1-1/p)$,
 and let $\Omega_1, \Omega_2\in C^\infty(\RR^d\backslash\{0\})$ be homogenous of degree $-\gamma_1$ and $-\gamma_2$ respectively.
Then
\begin{equation}\label{compositionp.tm.eq1}
U_{\Omega_1, p} (J_{\Omega_2} f)=J_{\Omega_1\Omega_2} f\quad {\rm for \ all} \ f\in {\mathcal S}.
\end{equation}
\end{Tm}

As a consequence of Theorems \ref{composition.tm} and \ref{compositionp.tm}, we have the following result about the left-inverse of the fractional Laplacian $(-\triangle)^{\gamma/2}$.

\begin{Cr} \label{leftinversefractionalderivative.cr}
Let $1\le p\le \infty$ and $\gamma>0$ satisfy either $1<p\le \infty$ or $p=1$ and $\gamma\not\in \ZZ_+$, and
the linear operator $I_{\gamma, p}$ be defined as in \eqref{fractionalderivative.veryolddef}. Then $I_{\gamma,p}$ is a left-inverse of the fractional Laplacian $(-\triangle)^{\gamma/2}$, i.e., $I_{\gamma, p} (-\triangle)^{\gamma/2} f=f$ for all $f\in {\mathcal S}$.
\end{Cr}

\begin{proof}
[\bf Proof of Theorem \ref{compositionp.tm}]
%
Let $k_1$ be the integral part of $\gamma_1-d(1-1/p)$. Then $-\gamma_2>k_1$ by the assumption.
Then ${\mathcal F}(J_{\Omega_2} f)(\xi)=\Omega_2(\xi) \hat f(\xi)$ and
 $\partial^{\bf i} ({\mathcal F}(J_{\Omega_2} f)(\xi))|_{\xi={\bf 0}}=0$
 for any ${\bf i}\in \ZZ_+$ with $|{\bf i}|\le  k_1$ and any Schwartz function $f$. This implies that
${\mathcal F}(U_{\Omega_1, p}(J_{\Omega_2} f))(\xi)$ is equal to
$$\Big(\widehat{J_{\Omega_2} f}(\xi)-\sum_{|{\bf i}|\le \gamma_1-d(1-1/p)}
\frac{\partial^{\bf i} ({\mathcal F}(J_{\Omega_2} f)(\xi))|_{\xi={\bf 0}}}{{\bf i}!} \xi^{\bf i}\Big)\Omega_1(\xi),$$
which is the same as ${\mathcal F}(J_{\Omega_1\Omega_2}f)(\xi)$. Hence the equation
\eqref{compositionp.tm.eq1} is established.
\end{proof}

%

\subsection{$L^p$-integrability in the spatial domain}
If  $\gamma\in (0, d(1-1/p))$, then it follows from  \eqref{newfractionalderivative.def} and Theorem \ref{generalizedrieszomega1.tm} that
$|U_{\Omega,p} f({\bf x})|\le C (1+\|{\bf x}\|)^{\gamma-d}, \ {\bf x}\in \RR^d$ (hence $U_{\Omega, p}f\in L^p$)
 for any Schwartz function $f$.
In this subsection, we provide a similar estimate for $U_{\Omega, p} f$
when $\gamma\ge d(1-1/p)$. 

%

\begin{Tm}
\label{iomegap.tm1}
  Let $0<\epsilon<1, 1\le p\le \infty, \gamma\in [d(1-1/p), \infty)\backslash \ZZ$,
   $k_1$ be the  integral part of $\gamma-d(1-1/p)$,
   and  $\Omega\in C^\infty(\RR^d\backslash \{\bf 0\})$ be
 a  homogeneous function of degree $-\gamma$.
If
\begin{equation} \label{iomegap.tm1.eq1}
|f({\bf x})| \le C (1+\|{\bf x}\|)^{-(k_1+1+d+\epsilon)}, \quad {\bf x}\in \RR^d\end{equation}
 then
 \begin{eqnarray}\label{iomegap.tm1.eq2}
 |U_{\Omega,p}f({\bf x})|  & \le &    C \Big(\sup_{{\bf z}\in \RR^d} |f({\bf z})| (1+\|{\bf z}\|)^{k_1+1+d+\epsilon}\Big)\nonumber\\
 & & \times
   \|{\bf x}\|^{\min(\gamma-k_1-d,0)} (1+\|{\bf x}\|)^{\max(\gamma-k_1-d,0)-1}
 \end{eqnarray}
for all ${\bf x}\in {\mathbb R}^d$,  and
  \begin{eqnarray}\label{iomegap.tm1.eq3}
 |U_{\Omega,p}f({\bf x})-U_{\Omega,p}f({\bf x}')|  & \le &    C \Big(\sup_{{\bf z}\in \RR^d} |f({\bf z})| (1+\|{\bf z}\|)^{k_1+1+d+\epsilon}\Big)\|{\bf x}-{\bf x}'\|^\delta \nonumber\\
 & & \times
   \|{\bf x}\|^{\min(\gamma-k_1-d-\delta,0)} (1+\|{\bf x}\|)^{\max(\gamma-k_1-d-\delta,0)-1}
 \end{eqnarray}
 for all ${\bf x}, {\bf x}'\in {\mathbb R}^d$ with $\|{\bf x}-{\bf x}'\|\le \|{\bf x}\|/4$, where $\delta<\min (|\gamma-k_1-d|, \epsilon)$.
 \end{Tm}

As an easy consequence of  Theorem \ref{iomegap.tm1}, we have

\begin{Cr}\label{iomegap.cor1}
 Let $1\le p\le \infty, \gamma\ge d(1-1/p)$,
   and  $\Omega\in C^\infty(\RR^d\backslash \{\bf 0\})$ be
 a  homogeneous function of degree $\gamma$.
 If both $\gamma$ and  $\gamma-d(1-1/p)$ are not nonnegative integers, then
  $U_{\Omega,p} f$ is H\"older continuous on $\RR^d\backslash \{\bf 0\}$ and belong to $L^p$ for any Schwartz function $f$.
\end{Cr}

\begin{proof}[Proof of Theorem \ref{iomegap.tm1}] We investigate three cases to establish the estimates in  \eqref{iomegap.tm1.eq2} and \eqref{iomegap.tm1.eq3}.

{\em Case I: $k_1+1-\gamma<0$}.\quad
Set $h_\xi(t)=\hat f(t\xi)$. Applying Taylor's expansion to the function $h_{\xi}$ gives
\begin{eqnarray}\label{iomegap.tm1.pf.eq1}
\hat f(\xi) & = &  h_\xi(1)=\sum_{s=0}^{k_1} \frac{h^{(s)}(0)}{s!}+\frac{1}{k_1!}\int_0^1 h_\xi^{(k_1+1)}(t) (1-t)^{k_1} dt\nonumber\\
&= & \sum_{|{\bf i}|\le k_1} \frac{\partial^{\bf i} \hat f({\bf 0})}{{\bf i}!} \xi^{\bf i}+ (k_1+1)
\sum_{|{\bf j}|=k_1+1} \frac{\xi^{\bf j}}{{\bf j}!} \int_0^1 \partial^{\bf j} \hat f(t\xi) (1-t)^{k_1} dt.
\end{eqnarray}
Hence
\begin{equation}\label{iomegap.tm1.pf.eq2}
\Big(\hat f(\xi)-\sum_{|{\bf i}|\le k_1} \frac{\partial^{\bf i} \hat f({\bf 0})}{{\bf i}!} \xi^{\bf i}\Big) \Omega(\xi)=
\sum_{|{\bf j}|=k_1+1} \frac{1}{{\bf j}!}\Omega_{\bf j}(\xi)  \widehat g_{\bf j}(\xi),
\end{equation}
where $\Omega_{\bf j}(\xi)=(i\xi)^{\bf j} \Omega(\xi)$ and
\begin{equation}\label{iomegap.tm1.pf.eq3}
g_{\bf j}({\bf x})=(k_1+1) \int_0^1 (1-t)^{k_1} (- {\bf x}/t)^{\bf j} f({\bf x}/t) t^{-d} dt\in L^1, \quad |{\bf j}|=k_1+1.
\end{equation}
Taking inverse Fourier transform at both sides of the equation \eqref{iomegap.tm1.pf.eq2} yields
\begin{equation}\label{iomegap.tm1.pf.eq4}
U_{\Omega, p} f({\bf x})=\sum_{|{\bf j}|=k_1+1}\frac{1}{{\bf j}!}
\int_{\RR^d} K_{\bf j}({\bf x}-{\bf y}) g_{\bf j}({\bf y}) d{\bf y}.
\end{equation}
where $K_{\bf j}, |{\bf j}|=k_1+1$, is the inverse Fourier transform of $\Omega_{\bf j}$.
Therefore
\begin{eqnarray}\label{iomegap.tm1.pf.eq5}
|U_{\Omega, p} f({\bf x})| &\le &  C \int_0^1\int_{\RR^d} \|{\bf x}-{\bf y}\|^{\gamma-d-k_1-1} \|{\bf y}/t\|^{k_1+1} |f({\bf y}/t)|  t^{-d} d{\bf y} dt\nonumber\\
& = & C \int_0^1\int_{\RR^d} \|{\bf x}-t {\bf y}\|^{\gamma-d-k_1-1} \|{\bf y}\|^{k_1+1} |f({\bf y})|   d{\bf y} dt\nonumber\\
&\le & C \Big (\sup_{{\bf z}\in \RR^d} |f({\bf z})| (1+\|{\bf z}\|)^{k_1+1+d+\epsilon}\Big)  \int_0^1 (t+\|{\bf x}\|)^{\gamma-d-k_1-1} dt\nonumber\\
&\le  & C \Big (\sup_{{\bf z}\in \RR^d} |f({\bf z})| (1+\|{\bf z}\|)^{k_1+1+d+\epsilon}\Big)
\nonumber\\
& & \times  \|{\bf x}\|^{\min(\gamma-d-k_1,0)}
(1+\|{\bf x}\|)^{\max(\gamma-d-k_1,0)-1},
\end{eqnarray}
where
the first inequality holds because
$K_{\bf j}\in C^\infty(\RR^d\backslash \{{\bf 0}\})$ is homogeneous of degree $\gamma-d-k_1-1\in (-d, 0)$ \cite[Theorems 7.1.16 and 7.1.18]{hormanderbook}, and the second inequality follows from \eqref{generalizedrieszomega1.tm.pf.eq3}.
Similarly,
\begin{eqnarray}\label{iomegap.tm1.pf.eq5+}
& & |U_{\Omega, p} f({\bf x})-U_{\Omega, p} f({\bf x}')|\nonumber\\
 &\le &  C \sum_{|{\bf j}|=k_1+1}
\int_{\|{\bf x}-{\bf y}\|\ge 2\|{\bf x}-{\bf x}'\|} \|{\bf x}-{\bf x}'\|^\delta
\|{\bf x}-{\bf y}\|^{\gamma-d-k_1-1-\delta} |g_{\bf j}({\bf y})| d{\bf y}\nonumber\\
& & + C \sum_{|{\bf j}|=k_1+1}
\int_{\|{\bf x}-{\bf y}\|\le 2\|{\bf x}-{\bf x}'\|} \|{\bf x}-{\bf y}\|^{\gamma-d-k_1-1} |g_{\bf j}({\bf y})| d{\bf y}\nonumber\\
& & + C \sum_{|{\bf j}|=k_1+1}
\int_{\|{\bf x}-{\bf y}\|\le 2\|{\bf x}-{\bf x}'\|} \|{\bf x}'-{\bf y}\|^{\gamma-d-k_1-1} |g_{\bf j}({\bf y})| d{\bf y}\nonumber\\
&\le  & C \Big (\sup_{{\bf z}\in \RR^d} |f({\bf z})| (1+\|{\bf z}\|)^{k_1+1+d+\epsilon}\Big)\|{\bf x}-{\bf x}'\|^\delta\nonumber\\
& & \times  \|{\bf x}\|^{\min(\gamma-d-k_1-\delta,0)}
(1+\|{\bf x}\|)^{\max(\gamma-d-k_1-\delta,0)-1}
\end{eqnarray}
for all ${\bf x}, {\bf x}'\in \RR^d$ with $\|{\bf x}-{\bf x}'\|\le \|{\bf x}\|/4$, where $\delta<\min(\epsilon, |\gamma-k_1-d|)$.
Then the desired estimate \eqref{iomegap.tm1.eq2}  and \eqref{iomegap.tm1.eq3} follow   from
\eqref{iomegap.tm1.pf.eq5} and \eqref{iomegap.tm1.pf.eq5+}
for  the case  $k_1+1-\gamma<0$.

{\em Case II:  $k_1+1-\gamma>0$ and $k_1\ge 1$. }\quad  
Applying Taylor's expansion to the function $h_{\xi}(t)=\hat f(t\xi)$, we have
\begin{equation*}
\label{iomegap.tm1.pf.eq7} \hat f(\xi) -\sum_{|{\bf i}|\le k_1} \frac{\partial^{\bf i} \hat f({\bf 0})}{{\bf i}!} \xi^{\bf i}
 =  k_1
\sum_{|{\bf j}|=k_1} \frac{\xi^{\bf j}}{{\bf j}!} \int_0^1 \big(\partial^{\bf j} \hat f(t\xi)-\partial^{{\bf j}}\hat f({\bf 0})\big) (1-t)^{k_1-1} dt.
\end{equation*}
Multiplying by $\Omega(\xi)$  both sides of the above equation and then taking the inverse Fourier transform, we obtain
\begin{equation}\label{iomegap.tm1.pf.eq8}
U_{\Omega, p} f({\bf x})= \sum_{|{\bf j}|=k_1}\frac{1}{{\bf j}!}
\Big(\int_{\RR^d} K_{\bf j}({\bf x}-{\bf y}) g_{\bf j}({\bf y}) d{\bf y}-K_{\bf j}({\bf x})\int_{\RR^d} g_{\bf j}({\bf y}) d{\bf y}\Big),
\end{equation}
where
\begin{equation}\label{iomegap.tm1.pf.eq8b}
g_{\bf j}({\bf x})=k_1 \int_0^1 (1-t)^{k_1-1} (- {\bf x}/t)^{\bf j} f({\bf x}/t) t^{-d} dt\in L^1, \quad |{\bf j}|=k_1.
\end{equation}
Recalling that $K_{\bf j}\in C^\infty(\RR^d\backslash \{{\bf 0}\}), |{\bf j}|=k_1$ are homogeneous of degree $\gamma-d-k_1\in (-d, 0)$,
\begin{equation}\label{iomegap.tm1.pf.eq9}
|\partial^{\bf i} K_{\bf j}({\bf x})|\le C \|{\bf x}\|^{\gamma-d-k_1-|{\bf j}|}, \quad |{\bf i}|\le 1.
\end{equation}
Combining \eqref{generalizedrieszomega1.tm.pf.eq3}, \eqref{iomegap.tm1.pf.eq8}, \eqref{iomegap.tm1.pf.eq8b} and \eqref{iomegap.tm1.pf.eq9}, we get
\begin{eqnarray} \label{iomegap.tm1.pf.eq10}
|U_{\Omega, p} f({\bf x})|
\!\!& \le &\!\! C \sum_{|{\bf j}|=k_1} \int_0^1 \int_{\RR^d} |K_{\bf j}({\bf x}-t{\bf y})-K_{\bf j}({\bf x})| \|{\bf y}\|^{k_1} |f({\bf y})| d{\bf y}\nonumber\\
\!\!& \le &\!\! C \Big(\sup_{{\bf z}\in \RR^d} |f({\bf z})| (1+\|{\bf z}\|)^{k_1+d+1+\epsilon}\Big)\nonumber\\
\!\! & & \!\!\times \Big\{\int_{0}^1 \int_{\|{\bf y}\|\le \|{\bf x}\|/2} t\|{\bf y}\| \|{\bf x}\|^{\gamma-d-k_1-1} (1+\|{\bf y}\|)^{-d-1-\epsilon} d{\bf y} dt\nonumber\\
\!\! & &\!\! + (1+\|{\bf x}\|)^{-1}
 \int_0^1 \int_{\|{\bf y}\|\ge \|{\bf x}\|/2} \|{\bf x}-t{\bf y}\|^{\gamma-d-k_1} (1+\|{\bf y}\|)^{-d-\epsilon} d{\bf y} dt\nonumber\\
\!\! & &\!\! + \|{\bf x}\|^{\gamma-d-k_1}\int_0^1
  \int_{\|{\bf y}\|\ge \|{\bf x}\|/2}  (1+\|{\bf y}\|)^{-d-1-\epsilon} d{\bf y} dt\Big\}\nonumber\\
\!\! & \le &  \!\! C\Big (\sup_{{\bf z}\in \RR^d} |f({\bf z})| (1+\|{\bf z}\|)^{k_1+d+1+\epsilon}\Big)\nonumber\\
\!\!& & \!\! \times
\|{\bf x}\|^{\min(\gamma-k_1-d, 0)} (1+\|{\bf x}\|)^{\max(\gamma-k_1-d, 0)-1},
\end{eqnarray}
and
\begin{eqnarray} \label{iomegap.tm1.pf.eq10+}
& & |U_{\Omega, p} f({\bf x})-U_{\Omega, p} f({\bf x}') |\nonumber\\
\!\!& \le &\!\! C \sum_{|{\bf j}|=k_1} \int_0^1 \Big(\int_{\|t{\bf y}\|\le \|{\bf x}\|/4}+\int_{\|t{\bf y}\|\ge 4\|{\bf x}\|}
+\int_{\|{\bf x}\|/4\le \|t{\bf y}\|\le 4\|{\bf x}\|}
\Big) \nonumber\\
& & \quad |K_{\bf j}({\bf x}-t{\bf y})-K_{\bf j}({\bf x})-
K_{\bf j}({\bf x}'-t{\bf y})+K_{\bf j}({\bf x}')| \|{\bf y}\|^{k_1} |f({\bf y})| d{\bf y}\nonumber\\
\!\!& \le &\!\! C \Big(\sup_{{\bf z}\in \RR^d} |f({\bf z})| (1+\|{\bf z}\|)^{k_1+d+1+\epsilon}\Big)\sum_{|{\bf j}|=k_1}
\Big\{ \|{\bf x}-{\bf x}'\|^\delta\nonumber\\
\!\! & & \!\!\times
\int_{0}^1 \int_{\|t{\bf y}\|\le \|{\bf x}\|/4} t\|{\bf y}\| \|{\bf x}\|^{\gamma-d-k_1-1-\delta} (1+\|{\bf y}\|)^{-d-1-\epsilon} d{\bf y} dt + \|{\bf x}-{\bf x}'\|^\delta \nonumber\\
\!\! & &\!\! \times
\int_0^1 \int_{t\|{\bf y}\|\ge 4\|{\bf x}\|}\big(\|{\bf x}\|^{\gamma-k_1-d-\delta} +\|{\bf y}\|^{\gamma-k_1-d-\delta}\big)
(1+\|{\bf y}\|)^{-d-1-\epsilon} d{\bf y} dt\nonumber\\
\!\! & &\!\! +
 \int_0^1 \int_{\|{\bf x}\|/4\le \|t{\bf y}\|\le 4 \|{\bf x}\|}
 \Big( | K_{\bf j}({\bf x}-t{\bf y})- K_{\bf j}({\bf x}'-t{\bf y})|+|K_{\bf j}({\bf x})-K_{\bf j}({\bf x})|\Big)\nonumber\\
  & & (1+\|{\bf x}\|/t)^{-d-1-\epsilon} d{\bf y} dt\Big\}\nonumber\\
\qquad \!\! & \le &  \!\! C\Big (\sup_{{\bf z}\in \RR^d} |f({\bf z})| (1+\|{\bf z}\|)^{k_1+d+1+\epsilon}\Big)\|{\bf x}-{\bf x}'\|^\delta
\|{\bf x}\|^{\gamma-k_1-d-\delta} (1+\|{\bf x}\|)^{-1}.
\end{eqnarray}
Then the desired estimates \eqref{iomegap.tm1.eq2} and \eqref{iomegap.tm1.eq3} are proved  in the case that  $k_1+1-\gamma>0$ and $k_1\ge 1$.

{\em Case III:  $k_1+1-\gamma>0$ and $k_1=0$.}\quad  In this case,  $\gamma\in (0, 1)$ and
\begin{equation}
U_{\Omega,p} f({\bf x})=\int_{\RR^d}\big( K({\bf x}-{\bf y})-K({\bf x})\big) f({\bf y}) d{\bf y}
\end{equation}
where $K$ is the inverse Fourier transform of $\Omega(\xi)$. Then, by applying the argument used in establishing \eqref{iomegap.tm1.pf.eq10}, we have
\begin{eqnarray}
|U_{\Omega, p} f({\bf x})|\!\! & \le  & \!\! C \Big(\sup_{{\bf z}\in \RR^d} |f({\bf z})| (1+\|{\bf z}\|)^{d+1+\epsilon}\Big)\nonumber\\
\!\! & & \!\!\times \Big\{ \int_{\|{\bf y}\|\le \|{\bf x}\|/2} t\|{\bf y}\| \|{\bf x}\|^{\gamma-d-1} (1+\|{\bf y}\|)^{-d-1-\epsilon} d{\bf y} \nonumber\\
\!\! & &\!\! + (1+\|{\bf x}\|)^{-1}
 \int_{\|{\bf y}\|\ge \|{\bf x}\|/2} \|{\bf x}-{\bf y}\|^{\gamma-d} (1+\|{\bf y}\|)^{-d-\epsilon} d{\bf y} \nonumber\\
\!\! & &\!\! + \|{\bf x}\|^{\gamma-d}
  \int_{\|{\bf y}\|\ge \|{\bf x}\|/2}  (1+\|{\bf y}\|)^{-d-1-\epsilon} d{\bf y} \Big\}\nonumber\\
\!\! &\le  & \!\! C\Big (\sup_{{\bf z}\in \RR^d} |f({\bf z})| (1+\|{\bf z}\|)^{d+1+\epsilon}\Big)
\|{\bf x}\|^{\gamma-d} (1+\|{\bf x}\|)^{-1},
\end{eqnarray}
and
\begin{eqnarray}
& & |U_{\Omega, p} f({\bf x})-U_{\Omega, p} f({\bf x}')|\nonumber\\
& \le &
\Big(\int_{\|{\bf y}\|\le \|{\bf x}\|/4}+\int_{\|{\bf y}\|\ge 4\|{\bf x}\|}
+\int_{\|{\bf x}\|/4\le \|{\bf y}\|\le 4\|{\bf x}\|}
\Big) \nonumber\\
& & \quad  |K({\bf x}-{\bf y})-K({\bf x})-K({\bf x'}-{\bf y})+K({\bf x}')| |f({\bf y}) | d{\bf y}\nonumber\\
& \le &
C\Big (\sup_{{\bf z}\in \RR^d} |f({\bf z})| (1+\|{\bf z}\|)^{k_1+d+1+\epsilon}\Big) \|{\bf x}-{\bf x}'\|^\delta \|{\bf x}\|^{\gamma-d-\delta} (1+\|{\bf x}\|)^{-1},
\end{eqnarray}
which yields the desired estimates \eqref{iomegap.tm1.eq2}  and \eqref{iomegap.tm1.eq3}  for $k_1+1-\gamma>0$ and $k_1=0$.
\end{proof}

\subsection{Unique dilation-invariant extension of the linear operator $i_\Omega$ with additional integrability in the spatial domain}
We now show that $U_{\Omega, p}$ is the only dilation-invariant extension of the linear operator $i_\Omega$ from the subspace ${\mathcal S}_\infty$
to the whole space ${\mathcal S}$ such that its image is contained in $L^p$.

  \begin{Tm}\label{time2.tm}
  Let $1\le p\le \infty$, $\gamma>0$ have the property that both $\gamma$ and $\gamma-d(1-1/p)$ are not nonnegative  integers,
   $\Omega\in C^\infty(\RR^d\backslash \{\bf 0\})$ be
 a nonzero homogeneous function of degree $-\gamma$, and
  the linear map $I$ from ${\mathcal S}$ to ${\mathcal S}'$
 be a homogeneous extension of the linear operator $i_\Omega$ on ${\mathcal S}_\infty$.
Then
 $If$ belongs to $L^p$  for any Schwartz function $f$ if and only if
  $I=U_{\Omega, p}$.

%
  \end{Tm}

\begin{proof}
The sufficiency follows from
 \eqref{fractionalderivative.tm1.pf.eq3} and
 Theorems \ref{generalizedriesz.tm}  and \ref{generalizedrieszomega1.tm} for $\gamma<d(1-1/p)$, and
 from
 \eqref{fractionalderivative.tm1.pf.eq3}, Theorem \ref{fractionalderivative.tm1}  and Corollary \ref{iomegap.cor1}
for $\gamma\ge d(1-1/p)$.
Now the necessity. By the assumption on the linear operator $I$ from ${\mathcal S}$ to ${\mathcal S}'$, similar to the argument used in  Lemma \ref{homogeneous1.lm}, we can find an integer $N$ and tempered distributions $H_{\bf i}, |{\bf i}|\le N$, such that
\begin{equation}\label{time2.tm.pf.eq3}
If=U_{\Omega, p} f+\sum_{|{\bf i}|\le N} \frac{\partial^{\bf i} \hat f({\bf 0})}{{\bf i}!} H_{\bf i} \quad {\rm for\ all} \ f\in {\mathcal S}.
\end{equation}
Replacing $f$ in \eqref{time2.tm.pf.eq3}  by $\psi_{\bf j}$ in \eqref{homogeneous2.lm.pf.eq7} and using \eqref{homogeneous2.lm.pf.eq8}
 gives that
$H_{\bf j}/{\bf j}!=I\psi_{\bf j}- U_{\Omega, p} \psi_{\bf j}$. Hence
\begin{equation}\label{time2.tm.pf.eq4} H_{\bf j}\in L^p\end{equation}
 by Corollary \ref{iomegap.cor1} and the assumption on the linear map $I$.
 By \eqref{time2.tm.pf.eq3}, Theorem \ref{fractionalderivative.tm1} and the assumption on the linear operator $I$,
  $(I-U_{\Omega,p})(\delta_t f)=t^{-\gamma} \delta_t ((I-U_{\Omega,p})f)$ for  all $f\in {\mathcal S}$.
Hence
$H_{\bf j}$ is  homogeneous of order $\gamma-d-|{\bf j}|$ by  Lemma \ref{homogeneous2.lm}.
This together with \eqref{time2.tm.pf.eq4} implies that $H_{\bf j}=0$ for all ${\bf j}\in \ZZ_+^d$ with $|{\bf j}|\le N$. The desired conclusion $I=U_{\Omega, p}$ then follows.
\end{proof}

\subsection{Unique dilation-invariant extension of the linear operator $i_\Omega$ with additional integrability in the Fourier domain}
In   this subsection, we characterize all those dilation-invariant extensions
$I$ of the linear operator $i_\Omega$ on the subspace ${\mathcal S}_\infty$ to the whole space ${\mathcal S}$
such that $\widehat{If}$ is $q$-integrable for any Schwartz function $f$.

 \begin{Tm}\label{iomega5.tm}
  Let $1\le q\le \infty, \gamma\in [d/q, \infty)\backslash \ZZ$ and  $\Omega\in C^\infty(\RR^d\backslash \{\bf 0\})$ be
 a nonzero homogeneous function of degree $-\gamma$, and
  the linear map $I$ from ${\mathcal S}$ to ${\mathcal S}'$
 be a dilation-invariant  extension of the linear operator $i_\Omega$ on ${\mathcal S}_\infty$.
 Then the following statements hold.

 \begin{itemize}

 \item[{(i)}] If $1\le q<\infty$, then the Fourier transform
  of $If$ belongs to $L^q$  for any Schwartz function $f$ if and only if
  $\gamma-d/q\not\in\ZZ_+$ and $I=U_{\Omega, q/(q-1)}$.

  \item[{(ii)}] If  $q=\infty$ and $\gamma\not\in \ZZ_+$, then
  the Fourier transform
  of $If$ belongs to $L^\infty$  for any Schwartz function $f$ if and only if
 $I=U_{\Omega, 1}$.

\item[{(iii)}] If $q=\infty$ and $\gamma\in \ZZ_+$, then
  the Fourier transform
  of $If$ belongs to $L^\infty$  for any Schwartz function $f$ if and only if
 \begin{equation}\label{iomega5.tm.eq1}
 \widehat{If}(\xi)  =  \widehat{U_{\Omega, 1}f}(\xi) +\sum_{|{\bf i}|=-\gamma}
  \frac{\partial^{\bf i}\hat f({\bf 0})}{ {\bf i}!} g_{\bf i}(\xi)
  \end{equation}
  for some bounded  homogeneous functions $g_{\bf i}, |{\bf i}|=-\gamma$, of degree $0$.
  \end{itemize}
  \end{Tm}

\begin{proof}
{\em (i)} \quad
The sufficiency follows from Theorem \ref{fractionalderivative.tm1} and Corollary \ref{fractionalderivative.cr1}.
Now we prove the necessity.
As every $q$-integrable function belong to $K_1$, similar to the argument used in the proof of Lemma \ref{homogeneous1.lm},
 we can find functions $g_{\bf i}\in K_1, |{\bf i}|\le N$,
 such that
\begin{eqnarray}\label{iomega5.tm.pf.eq1}
\widehat {If}(\xi) & = &  {\mathcal F}(U_{\Omega, q/(q-1)}f)(\xi)
+ \sum_{|{\bf i}|\le N} \frac{\partial^{\bf i} \hat f({\bf 0})}{{\bf i}!}  g_{\bf i}(\xi).
\end{eqnarray}
Let $\psi_{\bf j}, j\in \ZZ_+^d$ be defined as in \eqref{homogeneous2.lm.pf.eq7}.
Replacing $f$ by $\psi_{\bf j}$ with $|{\bf j}|\le N$ and using  \eqref{homogeneous2.lm.pf.eq8} gives
\begin{eqnarray}\label{iomega5.tm.pf.eq2}
\widehat {I\psi_{\bf j}}(\xi) & = &  \Big(\widehat{\psi_{\bf j}}(\xi) -\sum_{|{\bf i}|\le -\gamma-d/q} \frac{\partial^{\bf i} \hat \psi_{\bf j} ({\bf 0})}{{\bf i}!} \xi^{\bf i}\Big)\Omega(\xi)+ g_{\bf j}(\xi)\nonumber\\
& = & \left\{\begin{array}{ll}
\frac{\xi^{\bf j}}{{\bf j}!} (\phi(\xi)-1) \Omega(\xi)+ g_{\bf j}(\xi) &\quad  {\rm if} \ |{\bf j}|\le \gamma-d/q,\\
\frac{\xi^{\bf j}}{{\bf j}!} \phi(\xi) \Omega(\xi)+g_{\bf j}(\xi) & \quad  {\rm if}\ |{\bf j}|> \gamma-d/q.
\end{array}\right.
\end{eqnarray}
Note that $\frac{\xi^{\bf j}}{{\bf j}!} (\phi(\xi)-1) \Omega(\xi)\in L^q$
when $|{\bf j}|< \gamma-d/q$, and $\frac{\xi^{\bf j}}{{\bf j}!} \phi(\xi) \Omega(\xi)\in L^p$
when $|{\bf j}|>\gamma-d/q$.
This, together with \eqref{iomega5.tm.pf.eq2} and  the assumption that
$\widehat{I\psi_{\bf j}}\in L^q$, proves that
\begin{equation} \label{iomega5.tm.pf.eq3} g_{\bf j}\in L^q\quad
{\rm for\ all} \  {\bf j}\in \ZZ_+^d\quad {\rm with} \
 \gamma-d/q\ne |{\bf j}|\le N.
 \end{equation}
  By the homogeneous property of
the linear map $I$, the functions $g_{\bf i}, |{\bf i}|\le N$, are homogeneous of degree $-\gamma+|{\bf i}|$, i.e.,
\begin{equation}\label{iomega5.tm.pf.eq4}
g_{\bf i}(t\xi)= t^{-\gamma+|{\bf i}|} g_{\bf i}(\xi) , \quad \ {\rm for \ all} \ t>0.
\end{equation}
Combining \eqref{iomega5.tm.pf.eq3} and \eqref{iomega5.tm.pf.eq4}
 proves that $g_{\bf j}=0$  for all  ${\bf j}\in \ZZ_+^d$  with $\gamma-d/q\ne |{\bf j}|\le N$, and
the desired conclusion $\widehat {If}(\xi)  =   {\mathcal F}({U_{\Omega, q/(q-1)} f})(\xi)$
for all $f\in {\mathcal S}$
when $\gamma-d/q\not\in \ZZ_+$.

Now it suffices to  prove that $\gamma-d/q\not\in \ZZ_+$. Suppose on the contrary that
$\gamma-d/q\in \ZZ_+$. Then $1<q<\infty$ as $\gamma\not\in \ZZ$.
By \eqref{iomega5.tm.pf.eq2} and the assumption on the linear map $I$, we have
$$\int_{\xi\not\in {\rm supp} \phi}
|g_{\bf j}(\xi)-\xi^{\bf j} \Omega(\xi)/{\bf j}!|^q d\xi=
\int_{\xi\not\in {\rm supp} \phi} |\widehat {I\psi_{\bf j}}(\xi)|^q d\xi <\infty$$
for all ${\bf j}\in \ZZ_+^d$ with $|{\bf j}|=\gamma-d/q$.
This, together with \eqref{iomega5.tm.pf.eq4} and the fact that the support ${\rm supp} \phi$ of the function $\phi$ is
a bounded set, implies that $g_{\bf j}(\xi)-\xi^{\bf j} \Omega(\xi)/{\bf j}!=0$ for  all ${\bf j}\in \ZZ_+^d$ with $|{\bf j}|=\gamma-d/q$.
By substituting the above equality for  $g_{\bf j}$ into \eqref{iomega5.tm.pf.eq2} we obtain
\begin{equation}\label{iomega5.tm.pf.eq5}
\widehat{I\psi_{\bf j}}(\xi)= \phi(\xi) \xi^{\bf j} \Omega(\xi)/{\bf j}!\end{equation}
for all ${\bf j}\in \ZZ_+^d$ with $|{\bf j}|=\gamma-d/q$.
This leads to a contradiction,  as $\widehat{I\psi_{\bf j}}(\xi)\in L^q$
by the assumption on the linear map $I$, and
$\phi(\xi) \xi^{\bf j} \Omega(\xi)/{\bf j}!\not\in L^q$ by direction computation.

\bigskip{\em (ii) } and {\em (iii)}\quad
The necessity is true by   \eqref{iomega5.tm.eq1} and Theorem \ref{fractionalderivative.tm1}, while
the sufficiency follows from \eqref{iomega5.tm.pf.eq1} -- \eqref{iomega5.tm.pf.eq4}.
\end{proof}

\subsection{Proof of Theorem \ref{integrablefractionalderivative.tm}}
The conclusions in Theorem \ref{integrablefractionalderivative.tm} follow easily from \eqref{generalized.tm.pf.eq1}, \eqref{fractionalderivativeomegap.def}, Theorem \ref{time2.tm}
and Corollary \ref{leftinversefractionalderivative.cr}.

\section{Sparse Stochastic Processes}\label{poisson.section}

In this section, we will prove Theorem \ref{generalizedpoisson.tm} and fully characterize
 the generalized random process $P_\gamma w$, which is
 a solution of the stochastic partial differential equation \eqref{randompde.def}. In particular, we provide  its characteristic functional and its pointwise evaluation.

\subsection{Proof of Theorem \ref{generalizedpoisson.tm}}

To prove Theorem \ref{generalizedpoisson.tm}, we recall the Levy continuity theorem, and  a fundamental theorem
about the characteristic functional of a generalized random process.

\begin{Lm}\label{levy.lm}{\rm  (\cite{probabilitybook})}\
Let $\xi_k, k\ge 1$, be a sequence of  random variables whose characteristic functions are denoted by $\mu_k(t)$. If
$\lim_{k\to \infty} \mu_k(t)=\mu_{\infty}(t)$ for some continuous function $\mu_\infty(t)$ on the real line, then
$\xi_k$ converges to a random variable $\xi_\infty$ in distribution whose characteristic function ${\bf E}(e^{-it \xi_\infty})$
  is $\mu_\infty(t)$.
\end{Lm}

In the study of generalized random processes, the characteristic functional plays a similar role
to  the characteristic function of a random variable \cite{gelfandbook}. The idea is to  formally specify a generalized random process $\Phi$ by its
{\em characteristic functional} ${\mathcal Z}_\Phi$ given by
\begin{equation}
{\mathcal Z}_\Phi(f):={\bf E}(e^{-i\Phi(f)})=\int_{\RR} e^{-ix} dP({ x}), \quad f\in {\mathcal D},
\end{equation}
 where $P(x)$ denotes the probability that $\Phi(f)<x$.
For instance,  we can show (\cite{Unser2009})
 that the characteristic functional ${\mathcal Z}_w$ of the white Poisson noise \eqref{whitepoisson.def} is given by
 \begin{equation}
 {\mathcal Z}_w(f)=\exp\Big(\lambda \int_{\RR^d}\int_{\RR} \big(e^{-iaf({\bf x})}-1\big) dP(a)  d{\bf x} \Big) , \quad f\in {\mathcal D}.
 \end{equation}
The characteristic functional ${\mathcal Z}_\Phi$ of a generalized random process $\Phi$ is
 a functional from ${\mathcal D}$ to ${\mathbb C}$ that is continuous and positive-definite, and satisfies
${\mathcal Z}_\Phi(0)=1$. 
Here the {\em continuity} of a  functional $L$ from ${\mathcal D}$ to ${\mathbb C}$
 means that
 $\lim_{k\to \infty} L(f_k)=L(f)$ if $f_k\in {\mathcal D}$ tends to $f\in {\mathcal D}$ in the topology of the space ${\mathcal D}$, while
  a functional $L$ from ${\mathcal D}$ to ${\mathbb C}$ is said to be {\em positive-definite}
 if
\begin{equation}
 \sum_{j,k=1}^n L(f_j-f_k) c_j \bar c_k\ge 0
 \end{equation}
 for any $f_1, \ldots, f_n\in {\mathcal D}$ and any complex numbers $c_1, \ldots, c_n$. The remarkable aspect of the  theory of generalized random processes is that specification of ${\mathcal Z}_\Phi$ is sufficient to define a process in a consistent and unambiguous
  way. This is stated in  the
  fundamental  Minlos-Bochner theorem.
%
\begin{Tm}\label{gelfandbook.tm} {\rm (\cite{gelfandbook})}\
Let $L$ be a positive-definite continuous functional on ${\mathcal D}$ such that $L(0)=1$.
Then there exists a generalized random process $\Phi$ whose characteristic functional is $L$.
Moreover for any $f_1, \ldots, f_n\in {\mathcal D}$, we may take  the positive measure
$P({ x}_1, \ldots, { x}_n)$ as the distribution function of the random variable $\Phi(f_1), \ldots, \Phi(f_n)$,
where the Fourier transform of the positive measure $P({x}_1, \ldots, {x}_n)$ is $L(y_1 f_1+\cdots+y_nf_n)$, i.e.,
$$L(y_1f_1+\ldots+y_nf_n)=\int_{\RR^n} \exp(-i (x_1y_1+\ldots+x_ny_n)) dP(x_1, \ldots, x_n).$$
\end{Tm}


We are now ready to prove Theorem \ref{generalizedpoisson.tm}.

\begin{proof}[Proof of Theorem \ref{generalizedpoisson.tm}] Let $N\ge 1$ and $\varphi$ be a $C^\infty$ function supported in $B({\bf 0},2)$ and
taking the value one in $B({\bf 0}, 1)$.
For any $f\in {\mathcal D}$, define a sequence of random variables $\Phi_{\gamma, N}(f)$ associated with $f$ by
\begin{equation}\label{generalizedpoisson.tm.pf.eq1}
\Phi_{\gamma, N}(f):=\sum_{k} a_k \varphi({\bf x}_k/N) I_{\gamma, 1} f({\bf x}_k),
\end{equation}
 where the $a_k$'s are i.i.d. random variables with probability distribution $P(a)$, and where  the ${\bf x}_k$'s
are random point locations in ${\mathbb R}^n$ which are mutually independent and follow a spatial Poisson distribution
with Poisson parameter $\lambda>0$.
We will show that $\Phi_{\gamma, N}, N\ge 1$, define a sequence of generalized random processes,
whose limit  $P_\gamma w(f):=\sum_{k} a_k I_{\gamma, 1}(f)({\bf x}_k)$  is a  solution of the stochastic partial differential equation
\eqref{randompde.def}.

As $\varphi$ is a continuous function supported on $B({\bf 0}, 2)$,
\begin{equation}\label{generalizedpoisson.tm.pf.eq2}
\Phi_{\gamma, N}(f)=
\sum_{{\bf x}_k\in B({\bf 0}, 2N)} a_k \varphi({\bf x}_k/N) I_{\gamma, 1} f({\bf x}_k).
\end{equation}
Recall that $I_{\gamma, 1} f$ is  continuous on  $\RR^d\backslash \{\bf 0\}$ by Corollary \ref{iomegap.cor1}.
Then the summation of the right-hand side of \eqref{generalizedpoisson.tm.pf.eq2}
is well-defined whenever there are  finitely many ${\bf x}_k$ in $B({\bf 0}, 2N)$ with none of them belonging to
 $B({\bf 0}, \epsilon), \epsilon>0$. Note that
the probability that at least one of ${\bf x}_k$ lies in the small neighbor $B({\bf 0}, \epsilon)$ 
is  equal to  $$\sum_{n=1}^\infty
 e^{-\lambda |B({\bf 0}, \epsilon)|} \frac{(\lambda |B({\bf 0}, \epsilon)|)^n}{ n!}
 =1-e^{-\lambda |B({\bf 0}, \epsilon)|}\to 0 \quad {\rm as} \ \epsilon\to 0.$$
We then conclude that $\Phi_{\gamma, N}(f)$ is well-defined and $\Phi_{\gamma, N}(f)<\infty$ with probability one.

Denote  the characteristic function of the random variable $\Phi_{\gamma, N}(f)$ by $E_{\gamma, N, f}(t)$:
$$E_{\gamma, N, f}(t)= {\bf E}(e^{-it\Phi_{\gamma,N}(f)})={\bf E}(e^{-i\Phi_{\gamma, N}(tf)}).$$
Applying the same technique as in \cite[Appendix B]{tafti2009}, we  can show that
\begin{equation}\label{generalizedpoisson.tm.pf.eq3}
E_{\gamma, N, f}(t)=\exp\Big(\int_{\RR^d} \int_{\RR} \big(e^{-i a t \varphi({\bf x}/N) I_{\gamma, 1}f({\bf x})}-1\big) dP(a) d{\bf x}\Big).
\end{equation}
Moreover, the functional $E_{\gamma, N, f}(t)$ is continuous about $t$ by the dominated convergence theorem, because
 $$\Big|e^{-i a t \varphi({\bf x}/N) I_{\gamma, 1}f({\bf x})}-1\Big|\le |a| |t| |I_{\gamma, 1}f({\bf x})|$$ and
 $$\int_{\RR^d}\int_{\RR} |a| |I_{\gamma, 1}f({\bf x})| dP(a) d{\bf x}=
\Big(\int_{\RR} |a| dP(a)\Big)\times \Big(\int_{\RR^d} |I_{\gamma, 1}f({\bf x})| d{\bf x}\Big)<\infty$$
by Corollary \ref{iomegap.cor1} and the assumption on the distribution $P$.

Clearly the random  variable
$\Phi_{\gamma, N}(f)$ is linear about $f\in {\mathcal D}$;
i.e.,
\begin{equation}\label{generalizedpoisson.tm.pf.eq4}
\Phi_{\gamma,N}(\alpha f+\beta g)=\alpha \Phi_{\gamma,N}(f)+\beta\Phi_{\gamma,N}(g) \quad\ {\rm for \ all}\  f, g\in {\mathcal D} \ {\rm and} \ \alpha, \beta\in \RR.
\end{equation}
For any sequence of functions $f_k$ in ${\mathcal D}$ that converges to $f_\infty$ in the topology of ${\mathcal D}$,
it follows from Theorem \ref{iomegap.tm1} and Corollary \ref{iomegap.cor1} that
$\lim_{k\to \infty} \|I_{\gamma, 1} f_k-I_{\gamma, 1} f_\infty\|_1=0$.
Therefore
\begin{eqnarray}\label{generalizedpoisson.tm.pf.eq5}
& & \Big|\int_{\RR^d} \int_{\RR} \big(e^{-i a t\varphi({\bf x}/N) I_{\gamma, 1}f_k({\bf x})}-1\big) dP(a) d{\bf x}\nonumber\\
& & \quad -
\int_{\RR^d} \int_{\RR} \big(e^{-i a t \varphi({\bf x}/N) I_{\gamma, 1}f_\infty({\bf x})}-1\big) dP(a) d{\bf x}\Big|\nonumber\\
& \le & |t|\Big(\int_{\RR} |a| dP(a)\Big)\Big( \int_{\RR^d} \varphi({\bf x}/N) |I_{\gamma, 1}f_k({\bf x})-I_{\gamma, 1}f_\infty({\bf x})| d{\bf x}\Big)\nonumber\\
& \to & 0 \quad {\rm as} \  k\to \infty,
\end{eqnarray}
which implies that the characteristic function of $\Phi_{\gamma,N}(f_k)$ converges to the continuous characteristic  function of $\Phi_{\gamma,N}(f_\infty)$. Hence
the random variable $\Phi_{\gamma, N}(f_k)$ converges to $\Phi_{\gamma, N}(f_\infty)$ by Lemma \ref{levy.lm}, which in turn implies that
$\Phi_{\gamma, N}$ is continuous on ${\mathcal D}$.

Set
\begin{equation}L_{\gamma, N}(f)=E_{\gamma, N, f}(1).\end{equation}
For any sequence $c_l, 1\le l\le n$, of complex numbers and  $f_l, 1\le l\le n$, of functions in ${\mathcal D}$,
\begin{eqnarray} \label{generalizedpoisson.tm.pf.eq6}
\sum_{1\le l, l'\le n} L_{\gamma, N}(f_l-f_{l'})c_l \overline{c_{l'}} & = &
{\bf E}\Big(\sum_{l, l'=1}^n e^{-i \Phi_{\gamma,N}(f_l-f_{l'})} c_l\overline{c_{l'}}\Big)\nonumber\\
& =& {\bf E}\Big(\Big|\sum_{l=1}^n c_l
e^{-i \Phi_{\gamma, N}(f_l)}\Big|^2\Big) \ge 0,
\end{eqnarray}
which implies that $L_{\gamma, N}$ is positive-definite.
By Theorem \ref{gelfandbook.tm}, we conclude that $\Phi_{\gamma, N}$ defines a
generalized random process with characteristic functional $L_{\gamma, N}$.

Now we consider the limit of the above family of  generalized random processes $\Phi_{\gamma, N}, N\ge 1$.
By Corollary \ref{iomegap.cor1}, $I_{\gamma, 1} f$ is integrable for all $f\in {\mathcal D}$. Then
\begin{equation} \label{generalizedpoisson.tm.pf.eq7}
\lim_{N\to +\infty} E_{\gamma, N, f}(t)=\exp\Big(\int_{\RR^d} \int_{\RR} (e^{-i a t I_{\gamma, 1}f({\bf x})}-1) dP(a) d{\bf x}\Big)=:
E_{\gamma,  f}(t).
\end{equation}
Clearly  $E_{\gamma, f}(0)=1$ and  $E_{\gamma, f}(t)$ is continuous as $I_{\gamma,1}(f)$ is integrable.
Therefore by Lemma \ref{levy.lm}, $\Phi_{\gamma, N}(f)$ converges to a random variable, which is denoted by
$P_\gamma(f):=\sum_{k} a_k I_{\gamma, 1} f({\bf x}_k)$, in distribution.

As $I_{\gamma, 1} f$ is a continuous map from ${\mathcal D}$ to $L^1$, then
$\lim_{k\to \infty} \|I_{\gamma,1} f_k-I_{\gamma,1} f_\infty\|_1=0$ whenever  $f_k$ converges to $f$ in ${\mathcal D}$.
Hence
\begin{eqnarray}\label{generalizedpoisson.tm.pf.eq8}
& & \Big|\int_{\RR^d} \int_{\RR} \big(e^{-i a t I_{\gamma, 1}f_k({\bf x})}-1\big) dP(a) d{\bf x}\nonumber\\
& & \quad -
\int_{\RR^d} \int_{\RR} \big(e^{-i a t  I_{\gamma, 1}f_\infty({\bf x})}-1\big) dP(a) d{\bf x}\Big|\nonumber\\
& \le & |t|\Big(\int_{\RR} |a| dP(a)\Big)\Big( \int_{\RR^d}  |I_{\gamma, 1}f_k({\bf x})-I_{\gamma, 1}f_\infty({\bf x})| d{\bf x}\Big)\nonumber\\
& \to & 0 \quad {\rm as} \  k\to \infty,
\end{eqnarray}
which implies that
the characteristic function of $P_\gamma(f_k)$ converges to
the characteristic function of  $P_\gamma(f_\infty)$ (which is also continuous), and hence $P_\gamma(f_k)$
converges to $P_\gamma(f_\infty)$ in distribution by Lemma \ref{levy.lm}.
From the above argument, we see that $P_\gamma(f)$ is  continuous about $f\in {\mathcal D}$.

Define
$L_{\gamma}(f)=E_{\gamma, f}(1)$.
From \eqref{generalizedpoisson.tm.pf.eq6} and \eqref{generalizedpoisson.tm.pf.eq7}, we see that
\begin{equation}
\sum_{1\le l, l'\le n} L_{\gamma}(f_l-f_{l'}) c_i\overline{c_{i'}}=
\lim_{N\to \infty}
\sum_{1\le l, l'\le n} L_{\gamma, N}(f_l-f_{l'}) c_i\overline{c_{l'}}\ge 0
\end{equation}
for any sequence  $c_l, 1\le l\le n$, of complex numbers and $f_l, 1\le l\le n$, of functions in ${\mathcal D}$.
Therefore by Theorem \ref{gelfandbook.tm}, $P_\gamma w$ defines a generalized random  process with its characteristic functional given by
\begin{equation}
{\mathcal Z}_{P_\gamma w}(f)=
\exp\Big(\int_{\RR^d} \int_{\RR} (e^{-i a  I_{\gamma, 1}f({\bf x})}-1) dP(a) d{\bf x}\Big).\end{equation}
\end{proof}

\subsection{Pointwise evaluation}
In this section, we consider the pointwise characterization of the generalized random process $P_\gamma w$.

\begin{Tm}\label{pointpoisson.tm}
Let $\gamma, \lambda, P(a), P_\gamma w$ be as in Theorem \ref{generalizedpoisson.tm},
 and $I_{\gamma, 1}$ be defined as in \eqref{fractionalderivative.veryolddef}.
  Then
  \begin{equation}
  \label{pointpoisson.tm.eq3}
  P_{\gamma} w({\bf y}_0):=\lim_{N\to \infty} P_{\gamma} w (g_{N, {\bf y}_0})
  \end{equation}
is a random variable for every ${\bf y}_0\in \RR^d$ whose characteristic  function
is given by
\begin{equation}  \label{pointpoisson.tm.eq4}
{\bf E}(e^{-i t P_\gamma w({\bf y}_0)})=
\exp\Big(\lambda \int_{\RR}\int_{\RR} \big(e^{-ia t H_{{\bf y}_0}({\bf x})}-1\big) d{\bf x} dP(a)\Big), t\in \RR,
\end{equation}
where $g\in {\mathcal D}$ satisfies $\int_{\RR^d} g({\bf x}) d{\bf x}=1$, $g_{N, {\bf y}_0}({\bf x})=N^d g(N({\bf x}-{\bf y}_0))$, and
\begin{equation}\label{pointpoisson.tm.eq5}
\widehat{H_{{\bf y}_0}}(\xi)=\Big (e^{i\langle {\bf y}_0, \xi\rangle}-\sum_{|{\bf i}|\le \gamma}  \frac{(i{\bf y}_0)^{\bf i} \xi^{\bf i}}{{\bf i}!}\Big) \|\xi\|^{-\gamma}.
\end{equation}
\end{Tm}

An interpretation is that the random variable $P_{\gamma} w({\bf y}_0)$ in \eqref{pointpoisson.tm.eq3}
and its characteristic function ${\bf E}(e^{-i t P_\gamma w({\bf y}_0)})$ in \eqref{pointpoisson.tm.eq4}
 correspond formally to setting $f=\delta(\cdot-{\bf y}_0)$ (the delta distribution) in \eqref{generalizedpoisson.tm.eq1}
 and \eqref{generalizedpoisson.tm.eq2}, respectively.

To prove Theorem \ref{pointpoisson.tm}, we need a technical lemma.

\begin{Lm}\label{generalizedpoisson.lm}
Let $\gamma$ be a positive non-integer  number, $g\in {\mathcal D}$ satisfy $\int_{\RR^d} g({\bf x}) d{\bf x}=1$, and $H_{{\bf y}_0}$ be defined in \eqref{pointpoisson.tm.eq5}.
Then
\begin{equation}\label{generalizedpoisson.lm.eq1}
\lim_{N\to \infty} \|I_{\gamma, 1} g_{N, {\bf y}_0}-H_{{\bf y}_0}\|_1=0
\end{equation}
for all ${\bf y}_0\in \RR^d$, where $g_{N,{\bf y}_0}({\bf x})= N^d g(N ({\bf x}-{\bf y}_0))$.
\end{Lm}

\begin{proof} Let $K_{\bf j}$ be the inverse Fourier transform of $(i\xi)^{\bf j} \|\xi\|^{-\gamma}$ and $k_1$ be the integral part of the positive non-integer number $\gamma$. Then
from the argument in the proof of Theorem \ref{iomegap.tm1},
\begin{equation}\label{generalizedpoisson.lm.pf.eq1}
H_{\bf y}({\bf x})=\left\{\begin{array}{ll}
\sum_{|{\bf j}|=k_1} \frac{k_1}{{\bf j}!} \int_0^1 (K_{\bf j}({\bf x}-t{\bf y})-K_{\bf j}({\bf x})) (-{\bf y})^{\bf j} (1-t)^{k_1-1}
dt & {\rm if} \ k_1\ge 1\\
K_{\bf 0}({\bf x}-{\bf y})-K_{\bf 0}({\bf x}) & {\rm if} \ k_1=0.
\end{array}\right.
\end{equation}
Therefore  for ${\bf y}_0\ne 0$,
\begin{eqnarray*}
& & \|I_{\gamma, 1} g_{N, {\bf y}_0}-H_{{\bf y}_0}\|_1 \nonumber\\
& \le & C \sum_{|{\bf j}|=k_1}
\int_{\RR^d} \int_0^1 \int_{\RR^d}
| (K_{\bf j}({\bf x}-t{\bf y})-K_{\bf j}({\bf x})) {\bf y}^{\bf j} \nonumber\\
& &\quad -
(K_{\bf j}({\bf x}-t{\bf y}_0)-K_{\bf j}({\bf x})) {\bf y}_0^{\bf j} |  |g_{N, {\bf y}_0}({\bf y})| d{\bf y}  dt d{\bf x}\nonumber\\
& \le & C \sum_{|{\bf j}|=k_1}
\int_{\RR^d} \int_0^1 \int_{\RR^d}
| K_{\bf j}({\bf x}-t{\bf y})-K_{\bf j}({\bf x}-t{\bf y}_0)| \| {\bf y}\|^{k_1}
|g_{N, {\bf y}_0}({\bf y})| d{\bf y}  dt d{\bf x}\nonumber\\
& & +C\sum_{|{\bf j}|=k_1}
\int_{\RR^d} \int_0^1 \int_{\RR^d}
|K_{\bf j}({\bf x}-t{\bf y}_0)-K_{\bf j}({\bf x})| |{\bf y}^{\bf j}- {\bf y}_0^{\bf j} |
 |g_{N, {\bf y}_0}({\bf y})| d{\bf y}  dt d{\bf x}\nonumber\\
 &\le &  C
 \int_0^1 \int_{\RR^d} (t\|{\bf y}-{\bf y}_0\|)^{\gamma-k_1}
( \| {\bf y}_0\|^{k_1} +\|{\bf y}-{\bf y}_0\|^{k_1})
|g_{N, {\bf y}_0}({\bf y})| d{\bf y}  dt \nonumber\\
 & &  + C
  \int_0^1 \int_{\RR^d} (t\|{\bf y}_0\|)^{\gamma-k_1}
( \| {\bf y}_0\|^{k_1-1}\|{\bf y}-{\bf y}_0\| +\|{\bf y}-{\bf y}_0\|^{k_1})
|g_{N, {\bf y}_0}({\bf y})| d{\bf y}  dt \nonumber\\
 &\to & 0 \quad {\rm as} \ N\to \infty
\end{eqnarray*}
if $k_1\ge 1$, and
\begin{eqnarray*}
& & \|I_{\gamma, 1} g_{N, {\bf y}_0}-H_{{\bf y}_0}\|_1 \nonumber\\
& \le &
\int_{\RR^d} \int_{\RR^d}
| K_{\bf 0}({\bf x}-{\bf y})-K_{\bf 0}({\bf x}-{\bf y}_0)| |g_{N, {\bf y}_0}({\bf y}) d{\bf y} d{\bf x}\nonumber\\
& \le &
\int_{\RR^d} \Big(\int_{\|{\bf x}-{\bf y}\|\ge 2 \|{\bf y}-{\bf y}_0\|}| K_{\bf 0}({\bf x}-{\bf y})-K_{\bf 0}({\bf x}-{\bf y}_0)|
d{\bf x}\nonumber\\
& &
+\int_{\|{\bf x}-{\bf y}\|\le 2 \|{\bf y}-{\bf y}_0\|}
| K_{\bf 0}({\bf x}-{\bf y})| +
| K_{\bf 0}({\bf x}-{\bf y}_0)| d{\bf x} \Big)  |g_{N, {\bf y}_0}({\bf y}) | d{\bf y} \nonumber\\
& \le & C N^d \int_{\RR^d} \|{\bf y}-{\bf y}_0\|^\gamma |g(N({\bf y}-{\bf y}_0))| d{\bf y}\nonumber\\
& = & C N^{-\gamma} \int_{\RR^d} \|{\bf z}\|^\gamma |g({\bf z})| d{\bf z}\to 0 \quad {\rm as} \ N\to 0,
\end{eqnarray*}
if  $k_1=0$. This shows that \eqref{generalizedpoisson.lm.eq1} for ${\bf y}_0\ne {\bf 0}$.

The limit  in \eqref{generalizedpoisson.lm.eq1} for ${\bf y}_0={\bf 0}$ can be proved by using a similar argument, the detail of which are omitted
 here.
\end{proof}

\begin{proof} [Proof of Theorem \ref{pointpoisson.tm}]
By Lemma \ref{generalizedpoisson.lm} and the dominated convergence theorem,
\begin{equation}
\lim_{N\to \infty}
\int_{\RR^d} \int_{\RR} (e^{-i at I_{\gamma, 1} g_{N, {\bf y}_0}({\bf x})}-1) dP(a) d{\bf x}=
\int_{\RR^d} \int_{\RR} (e^{-i at  H_{{\bf y}_0}({\bf x})}-1) dP(a) d{\bf x}
\end{equation}
for all $t\in \RR$. Moreover as $H_{{\bf y}_0}$ is integrable from Corollary \ref{iomegap.cor1} and
Lemma \ref{generalizedpoisson.lm},
the function
$\int_{\RR^d} \int_{\RR} (e^{-i at I_{\gamma, 1} H_{{\bf y}_0}({\bf x})}-1) dP(a) d{\bf x}$ is continuous about $t$. Therefore
\eqref{pointpoisson.tm.eq3} and \eqref{pointpoisson.tm.eq4} follows from Lemma \ref{levy.lm}.
\end{proof}


\noindent{\bf Acknowledgement.} {\rm This work was done when the first named author was visiting
Ecole Polytechnique Federale de Lausanne   on his sabbatical leave.
He would like to thank Professors Michael Unser and Martin Vetterli for the hospitality and  fruitful discussions. }

\end{document}